\documentclass[11pt,letterpaper]{article}

\usepackage{thmtools}
\usepackage{thm-restate}
\usepackage{bm}
\usepackage{mathrsfs}           
\usepackage{vmargin,fancyhdr}   

\usepackage{algorithm}
\usepackage{algorithmic}

\usepackage{subcaption}
\usepackage{enumerate}

\usepackage{amsmath,amssymb, amsthm}    
\usepackage{verbatim}           
\usepackage{xspace}             
\usepackage{graphicx,float}     
\usepackage{ifthen,calc}        
\usepackage{textcomp}           
\usepackage{fancybox}           
\usepackage{hhline}             
\usepackage{float}              

\usepackage[vflt]{floatflt}     
\usepackage[small,compact]{titlesec}
\usepackage{setspace}
\renewcommand{\baselinestretch}{1.0}

\usepackage{color}
\definecolor{ForestGreen}{rgb}{0.1333,0.5451,0.1333}
\usepackage{thumbpdf}
\usepackage[letterpaper,
colorlinks,linkcolor=ForestGreen,citecolor=ForestGreen,
backref,
bookmarks,bookmarksopen,bookmarksnumbered]
{hyperref}

\setpapersize{USletter}
\setmarginsrb{1in}{.5in}        
{1in}{.5in}        
{.25in}{.25in}     
{.25in}{.5in}      
\newcommand{\showccc}[0]{0}
\newcommand{\ccc}[2][nothing]{
	\ifthenelse{\showccc=0}{}{
		\ensuremath{^{\Lsh\Rsh}}\marginpar{\raggedright\tiny\textsf{%
				\ifthenelse{\equal{#1}{nothing}}{}{\textbf{#1}\\}#2}}}}
\newcounter{hours}\newcounter{minutes}
\newcommand{\hhmm}{%
	\setcounter{hours}{\time/60}%
	\setcounter{minutes}{\time-\value{hours}*60}%
	\ifthenelse{\value{hours}<10}{0}{}\thehours:%
	\ifthenelse{\value{minutes}<10}{0}{}\theminutes}
\ifthenelse{\showccc=0}{\rhead{}}{\rhead{\today \ [\hhmm]}}
\newtheorem{theorem}{Theorem}[section]
\newtheorem{claim}[theorem]{Claim}
\newtheorem{proposition}[theorem]{Proposition}
\newtheorem{corollary}{Corollary}
\newtheorem{definition}[theorem]{Definition}
\newtheorem{conjecture}[theorem]{Conjecture}
\newtheorem{remark}[theorem]{Remark}
\newtheorem{lemma}{Lemma}
\newtheorem{fact}[theorem]{Fact}
\newtheorem{example}[theorem]{Example}
\newtheorem{axiom}[theorem]{Axiom}
\newtheorem{property}[theorem]{Property}
\newtheorem{invariant}[theorem]{Invariant}

\newcommand{\Ohstar}{\ensuremath{\Oh^\star}}
\newcommand{\oh}{\ensuremath{o}}
\newcommand{\ohstar}{\ensuremath{\oh^\star}}
\newcommand{\bv}[1]{\mathbf{#1}}
\def\D{\mathrm{d}}
\newcommand{\dmax}{\norm{d}_{\infty}}

\newcommand{\QED}[0]{\hfill\ensuremath{\blacksquare}\medspace\\}
\newcommand{\blah}[0]{\ensuremath{\blacktriangleright
		\clubsuit\diamondsuit\heartsuit\spadesuit
		\blacktriangleleft}\xspace}
\newcommand{\set}[1]{\ensuremath{\{\,{#1}\,\}}}
\newcommand{\stm}[0]{\ensuremath{\,\boldsymbol{|}\,}}
\newcommand{\size}[1]{\ensuremath{|{#1}|}}
\newcommand{\etal}[0]{\emph{et~al.}\xspace}
\newcommand{\etc}[0]{\emph{etc.}\xspace}
\newcommand{\ie}[0]{\emph{i.e.}\xspace}
\newcommand{\eg}[0]{\emph{e.g.}\xspace}
\newcommand{\vs}[0]{\emph{vs.}\xspace}
\newcommand{\half}{\frac{1}{2}}

\newcommand{\cart}[2]{
	\ensuremath{\{{#1} {\boldsymbol\otimes} {#2}\}}}
\newcommand{\supports}[0]{\succcurlyeq}
\newcommand{\TT}[0]{
	\ensuremath{{}^{\!\normalfont\textsf{T}\!}}}
\newcommand{\IV}[0]{%
	\ensuremath{{}^{\normalfont\textsf{-}\!1\!}}}
\newcommand{\vv}[1]{\ensuremath{\boldsymbol{#1}}}
\newcommand{\sqrootinv}[1]{\ensuremath{%
		{#1}^{\textsf{-}\frac{1}{2}}}}
\newcommand{\sqroot}[1]{\ensuremath{{#1}^{\frac{1}{2}}}}
\newcommand{\R}[0]{\ensuremath{\mathbb{R}}}
\newcommand{\diag}[1]{\textbf{\textup{diag}}\left(#1\right)}
\newcommand{\reminder}[1]{{\textsf{\textcolor{red}{[#1]}}}}

\floatstyle{ruled}
\newfloat{algo}{tbp}{lop}
\floatname{algo}{Algorithm}

\DeclareMathOperator{\nullsp}{null}
\DeclareMathOperator{\spansp}{span}
\DeclareMathOperator{\congest}{cong}
\DeclareMathOperator{\dil}{dil}

\newcommand{\trace}[1]{\textbf{tr}\left[#1\right]}
\newcommand{\norm}[1]{\|#1\|}
\newcommand{\norms}[1]{\|#1\|}
\newcommand{\tr}[0]{\textup{Tr}}
\newcommand{\E}[0]{\mathbb{E}}
\newcommand{\opt}[0]{\textup{OPT}}
\newcommand{\topt}[0]{\tilde{\textup{OPT}}}
\newcommand{\prox}[0]{\textup{Prox}}
\newcommand{\kjtian}[1]{{\color{red}\textbf{kjtian:} #1 }}

\usepackage{float}
\floatstyle{plain}\newfloat{myfig}{t}{figs}[section]
\floatname{myfig}{\textsc{Figure}}
\floatstyle{plain}\newfloat{myalg}{H}{algs}[section]
\floatname{myalg}{}
\newcommand{\todo}[1]{{\color{red}\textbf{TODO:}#1}}

\newcommand{\onesv}{\textbf{1}}

\newcommand{\obj}{\mathcal{OBJ}}

\newcommand{\identity}{\textbf{I}}

\newcommand{\nbr}[1]{\left\|#1\right\|}
\newcommand{\rnorm}[1]{\nbr{#1}_\matr}
\newcommand{\rpnorm}[1]{\nbr{#1}_{\matr'}}
\newcommand{\expct}[3]{\ensuremath{\mathop{\text{\normalfont \textbf{E}}_{#1}}_{#2}}\left[#3\right]}
\newcommand{\prob}[1]{\ensuremath{\mathop{\text{\normalfont \textbf{P}}}}\left[#1\right]}

\newcommand{\epsmac}{\epsilon_{m}}
\newcommand{\vecd}{\textbf{d}}
\newcommand{\matcholesky}{\textbf{C}}

\newcommand{\tvect}[2]{
	\ensuremath{\Bigl(\negthinspace\begin{smallmatrix}#1\\#2\end{smallmatrix}\Bigr)}}

\newcommand{\rasmus}[1]{{\color{green} RASMUS: #1}}
\newcommand{\jakub}[1]{{\color{red} JAKUB: #1}}
\newcommand{\richard}[1]{{\color{blue} RICHARD: #1}}

\global\long\def\R{\mathbb{R}}
\global\long\def\Rn{\mathbb{R}^{n}}
\global\long\def\Rm{\mathbb{R}^{m}}
\global\long\def\Z{\mathbb{Z}}
\global\long\def\rPos{\R^{+}}
\global\long\def\rNonNeg{\R^{\geq0}}

\global\long\def\ellOne{\ell_{1}}
\global\long\def\ellTwo{\ell_{2}}
\global\long\def\ellInf{\ell_{\infty}}
\global\long\def\ellP{\ell_{p}}

\global\long\def\otilde{\tilde{\mathcal{O}}}

\global\long\def\boldVar#1{\mathbf{#1}}
\global\long\def\mvar#1{\boldVar{#1}}
\global\long\def\vvar#1{\vec{#1}}

\newcommand{\defeq}{\stackrel{\mathrm{{\scriptscriptstyle def}}}{=}}
\newcommand{\inprod}[2]{\langle#1, #2\rangle}
\newcommand{\inprodFull}[2]{\left\langle#1, #2\right\rangle}
\newcommand{\Oh}[1]{O\left(#1\right)}
\newcommand{\tildeOh}[1]{\tilde{O}\left(#1\right)}
\newcommand{\1}{\mathbf{1}}
\newcommand{\cmax}{\snorm{C}_{\max}}
\newcommand{\urc}{\mathcal{U}_{r, c}}

\global\long\def\gradient{\bigtriangledown}
\global\long\def\grad{\gradient}
\global\long\def\hessian{\gradient^{2}}
\global\long\def\hess{\hessian}
\global\long\def\jacobian{\mvar J}
\global\long\def\gradIvec#1{\vvar{f_{#1}}}
\global\long\def\gradIval#1{f_{#1}}

\global\long\def\onesVec{\vec{\mathbb{1}}}
\global\long\def\setVec#1{\onesVec_{#1}}
\global\long\def\indicVec#1{\onesVec_{#1}}

\global\long\def\trans#1{#1^{\intercal}}

\global\long\def\ith#1#2{{#1}^{(#2)}}

\global\long\def\specGeq{\succeq}
\global\long\def\specLeq{\preceq}
\global\long\def\specGt{\succ}
\global\long\def\specLt{\prec}
\global\long\def\ithEigValOf#1#2{\ith{\lambda}{#2}_{#1}}
\global\long\def\ithEigVecOf#1#2{\ith{\vvar v}{#2}_{#1}}

\global\long\def\innerProduct#1#2{\big\langle#1 , #2 \big\rangle}
\global\long\def\norm#1{\left\|#1\right\|}
\global\long\def\snorm#1{\|#1\|}
\global\long\def\normA#1{\norm{#1}_{\ma}}
\global\long\def\normInf#1{\norm{#1}_{\infty}}
\global\long\def\normOne#1{\norm{#1}_{1}}
\global\long\def\normTwo#1{\norm{#1}_{2}}
\global\long\def\normCoord#1#2{\norm{#1}_{(#2)}}
\global\long\def\normDual#1{{\norm{#1}^{*}}}

\global\long\def\partiali#1{{\partialiOp}_{#1}}

\global\long\def\dualVec#1{{#1}^{\#}}
\global\long\def\dualVecOne#1{{#1}^{\#_{1}}}
\global\long\def\dualVecInf#1{{#1}^{\#_{\infty}}}
\global\long\def\dualVecP#1{{#1}^{\#_{p}}}
\global\long\def\sharpv#1{\dualVec{#1}}
\global\long\def\sharpgrad#1{\dualVec{\gradient#1}}

\global\long\def\OPT{\mathrm{opt}}

\global\long\def\fopt{f^{*}}

\global\long\def\va{\vvar a}
\global\long\def\vb{\vvar b}
\global\long\def\vc{\vvar c}
\global\long\def\vd{\vvar d}
\global\long\def\ve{\vvar e}
\global\long\def\vf{\vvar f}
\global\long\def\vg{\vvar g}
\global\long\def\vh{\vvar h}
\global\long\def\vl{\vvar l}
\global\long\def\vm{\vvar m}
\global\long\def\vn{\vvar n}
\global\long\def\vo{\vvar o}
\global\long\def\vp{\vvar p}
\global\long\def\vq{\vvar q}
\global\long\def\vr{\vvar r}
\global\long\def\vs{\vvar s}
\global\long\def\vu{\vvar u}
\global\long\def\vv{\vvar v}
\global\long\def\vw{\vvar w}
\global\long\def\vx{\vvar x}
\global\long\def\vy{\vvar y}
\global\long\def\vz{\vvar z}

\global\long\def\vpi{\vvar{\pi}}
\global\long\def\vxi{\vvar{\xi}}
\global\long\def\vchi{\vvar{\chi}}
\global\long\def\valpha{\vvar{\alpha}}
\global\long\def\veta{\vvar{\eta}}
\global\long\def\vlambda{\vvar{\lambda}}
\global\long\def\vmu{\vvar{\mu}}
\global\long\def\vdelta{\vvar{\Delta}}
\global\long\def\vsigma{\vvar{\sigma}}
\global\long\def\vzero{\vvar 0}
\global\long\def\vones{\vvar 1}

\global\long\def\vhx{\hat{x}}
\global\long\def\vhy{\hat{y}}
\global\long\def\vhs{\hat{s}}
\global\long\def\vhc{\hat{c}}
\global\long\def\vbx{\bar{x}}
\global\long\def\vby{\bar{y}}
\global\long\def\vbs{\bar{s}}

\global\long\def\xopt{\vvar x^{*}}
\global\long\def\varVec{\vvar x}
\global\long\def\varVecA{\vvar x}
\global\long\def\varVecB{\vvar y}
\global\long\def\varMat{\mvar A}
\global\long\def\varMatA{\mvar A}
\global\long\def\varMatB{\mvar B}
\global\long\def\varSubedges{H}

\global\long\def\ma{\mvar A}
\global\long\def\mb{\mvar B}
\global\long\def\mc{\mvar C}
\global\long\def\md{\mvar D}
\global\long\def\mf{\mvar F}
\global\long\def\mg{\mvar G}
\global\long\def\mh{\mvar H}
\global\long\def\mI{\mvar I}
\global\long\def\mm{\mvar M}
\global\long\def\mq{\mvar Q}
\global\long\def\mr{\mvar R}
\global\long\def\ms{\mvar S}
\global\long\def\mt{\mvar T}
\global\long\def\mU{\mvar U}
\global\long\def\mv{\mvar V}
\global\long\def\mw{\mvar W}
\global\long\def\mx{\mvar X}
\global\long\def\my{\mvar Y}
\global\long\def\mz{\mvar Z}
\global\long\def\mproj{\mvar P}
\global\long\def\mSigma{\mvar{\Sigma}}
\global\long\def\mLambda{\mvar{\Lambda}}
\global\long\def\mha{\hat{\mvar A}}
\global\long\def\mzero{\mvar 0}
\global\long\def\mlap{\mvar{\mathcal{L}}}
\global\long\def\mpi{\mvar{\mathcal{\Pi}}}

\global\long\def\mdiag{\mvar{\texttt{d}iag}}

\global\long\def\paramLengths{\vvar l}
\global\long\def\paramCapacity{\vvar{\mu}}
\global\long\def\paramWeights{\vvar w}

\global\long\def\weightVec{\vvar w}
\global\long\def\convexSet{Q}
\global\long\def\convexSetElement{\vvar q}
\global\long\def\oracle{\mathcal{O}}
\global\long\def\moracle{\mvar O}
\global\long\def\oracleOf#1{\oracle\left(#1\right)}
\global\long\def\nSamples{s}
\global\long\def\simplex{\Delta}

\global\long\def\abs#1{\left|#1\right|}

\global\long\def\capacityMatrix{\mvar U}

\global\long\def\load{\mathrm{load}}
\global\long\def\rload{\mathrm{rload}}
\global\long\def\tpath{\mathrm{path}}
\global\long\def\cost{\mathrm{cost}}
\global\long\def\tr{\mathrm{tr}}

\global\long\def\timeNearlyOp{\tilde{\mathcal{O}}}
\global\long\def\timeNearlyLinear{\timeNearlyOp}

\global\long\def\varFun{f}

\global\long\def\funLip{L}
\global\long\def\funCon{\mu}

\global\long\def\stepOpt{T}
\global\long\def\stepOptCoordinate#1{\stepOpt_{(#1)}}

\global\long\def\reff#1{R_{#1}^{\text{eff}}}

\global\long\def\im{\mathrm{im}}

\global\long\def\infseq#1#2{\{#1\}_{#2 = 0}^{\infty}}

\global\long\def\ceil#1{\left\lceil #1 \right\rceil }

\global\long\def\runtime{\mathcal{T}}
\global\long\def\timeOf#1{\runtime\left(#1\right)}

\global\long\def\domain{\mathcal{D}}

\global\long\def\argmin{\mathrm{argmin}}
\global\long\def\nnz{\mathrm{nnz}}
\global\long\def\vol{\mathrm{vol}}
\global\long\def\supp{\mathrm{supp}}
\global\long\def\dist{\mathcal{D}}

\newcommand{\smax}{\textup{smax}}

\newcommand{\sidford}[1]{{\{\color{green}\textbf{sidford:} #1 \}}}

\usepackage{tikz}
\usetikzlibrary{positioning,chains,fit,shapes,calc}
  \usepackage{nth}
  \usepackage{intcalc}

  \newcommand{\cSTOC}[1]{\nth{\intcalcSub{#1}{1968}}\ Annual\ ACM\ Symposium\ on\ Theory\ of\ Computing\ (STOC)}
  \newcommand{\cFSTTCS}[1]{\nth{\intcalcSub{#1}{1980}}\ International\ Conference\ on\ Foundations\ of\ Software\ Technology\ and\ Theoretical\ Computer\ Science\ (FSTTCS)}
  \newcommand{\cCCC}[1]{\nth{\intcalcSub{#1}{1985}}\ Annual\ IEEE\ Conference\ on\ Computational\ Complexity\ (CCC)}
  \newcommand{\cFOCS}[1]{\nth{\intcalcSub{#1}{1959}}\ Annual\ IEEE\ Symposium\ on\ Foundations\ of\ Computer\ Science\ (FOCS)}
  \newcommand{\cRANDOM}[1]{\nth{\intcalcSub{#1}{1996}}\ International\ Workshop\ on\ Randomization\ and\ Computation\ (RANDOM)}
  \newcommand{\cISSAC}[1]{#1\ International\ Symposium\ on\ Symbolic\ and\ Algebraic\ Computation\ (ISSAC)}
  \newcommand{\cICALP}[1]{\nth{\intcalcSub{#1}{1973}}\ International\ Colloquium\ on\ Automata,\ Languages and\ Programming\ (ICALP)}
  \newcommand{\cCOLT}[1]{\nth{\intcalcSub{#1}{1987}}\ Annual\ Conference\ on\ Computational\ Learning\ Theory\ (COLT)}
  \newcommand{\cCSR}[1]{\nth{\intcalcSub{#1}{2005}}\ International\ Computer\ Science\ Symposium\ in\ Russia\ (CSR)}
  \newcommand{\cMFCS}[1]{\nth{\intcalcSub{#1}{1975}}\ International\ Symposium\ on\ the\ Mathematical\ Foundations\ of\ Computer\ Science\ (MFCS)}
  \newcommand{\cPODS}[1]{\nth{\intcalcSub{#1}{1981}}\ Symposium\ on\ Principles\ of\ Database\ Systems\ (PODS)}
  \newcommand{\cSODA}[1]{\nth{\intcalcSub{#1}{1989}}\ Annual\ ACM-SIAM\ Symposium\ on\ Discrete\ Algorithms\ (SODA)}
  \newcommand{\cNIPS}[1]{Advances\ in\ Neural\ Information\ Processing\ Systems\ \intcalcSub{#1}{1987} (NIPS)}
  \newcommand{\cWALCOM}[1]{\nth{\intcalcSub{#1}{2006}}\ International\ Workshop\ on\ Algorithms\ and\ Computation\ (WALCOM)}
  \newcommand{\cSoCG}[1]{\nth{\intcalcSub{#1}{1984}}\ Annual\ Symposium\ on\ Computational\ Geometry\ (SCG)}
  \newcommand{\cKDD}[1]{\nth{\intcalcSub{#1}{1994}}\ ACM\ SIGKDD\ International\ Conference\ on\ Knowledge\ Discovery\ and\ Data\ Mining\ (KDD)}
  \newcommand{\cICML}[1]{\nth{\intcalcSub{#1}{1983}}\ International\ Conference\ on\ Machine\ Learning\ (ICML)}
  \newcommand{\cAISTATS}[1]{\nth{\intcalcSub{#1}{1997}}\ International\ Conference\ on\ Artificial\ Intelligence\ and\ Statistics\ (AISTATS)}
  \newcommand{\cITCS}[1]{\nth{\intcalcSub{#1}{2009}}\ Conference\ on\ Innovations\ in\ Theoretical\ Computer\ Science\ (ITCS)}
  \newcommand{\cPODC}[1]{{#1}\ ACM\ Symposium\ on\ Principles\ of\ Distributed\ Computing\ (PODC)}
  \newcommand{\cAPPROX}[1]{\nth{\intcalcSub{#1}{1997}}\ International\ Workshop\ on\ Approximation\ Algorithms\ for\  Combinatorial\ Optimization\ Problems\ (APPROX)}
  \newcommand{\cSTACS}[1]{\nth{\intcalcSub{#1}{1983}}\ International\ Symposium\ on\ Theoretical\ Aspects\ of\  Computer\ Science\ (STACS)}
  \newcommand{\cMTNS}[1]{\nth{\intcalcSub{#1}{1991}}\ International\ Symposium\ on\ Mathematical\ Theory\ of\  Networks\ and\ Systems\ (MTNS)}
  \newcommand{\cICM}[1]{International\ Congress\ of\ Mathematicians\ {#1} (ICM)}

  \newcommand{\pSTOC}[1]{Preliminary\ version\ in\ the\ \cSTOC{#1}, #1}
  \newcommand{\pFSTTCS}[1]{Preliminary\ version\ in\ the\ \cFSTTCS{#1}, #1}
  \newcommand{\pCCC}[1]{Preliminary\ version\ in\ the\ \cCCC{#1}, #1}
  \newcommand{\pFOCS}[1]{Preliminary\ version\ in\ the\ \cFOCS{#1}, #1}
  \newcommand{\pRANDOM}[1]{Preliminary\ version\ in\ the\ \cRANDOM{#1}, #1}
  \newcommand{\pISSAC}[1]{Preliminary\ version\ in\ the\ \cISSAC{#1}, #1}
  \newcommand{\pICALP}[1]{Preliminary\ version\ in\ the\ \cICALP{#1}, #1}
  \newcommand{\pCOLT}[1]{Preliminary\ version\ in\ the\ \cCOLT{#1}, #1}
  \newcommand{\pCSR}[1]{Preliminary\ version\ in\ the\ \cCSR{#1}, #1}
  \newcommand{\pMFCS}[1]{Preliminary\ version\ in\ the\ \cMFCS{#1}, #1}
  \newcommand{\pPODS}[1]{Preliminary\ version\ in\ the\ \cPODS{#1}, #1}
  \newcommand{\pSODA}[1]{Preliminary\ version\ in\ the\ \cSODA{#1}, #1}
  \newcommand{\pNIPS}[1]{Preliminary\ version\ in\ \cNIPS{#1}, #1}
  \newcommand{\pWALCOM}[1]{Preliminary\ version\ in\ the\ \cWALCOM{#1}, #1}
  \newcommand{\pSoCG}[1]{Preliminary\ version\ in\ the\ \cSoCG{#1}, #1}
  \newcommand{\pKDD}[1]{Preliminary\ version\ in\ the\ \cKDD{#1}, #1}
  \newcommand{\pICML}[1]{Preliminary\ version\ in\ the\ \cICML{#1}, #1}
  \newcommand{\pAISTATS}[1]{Preliminary\ version\ in\ the\ \cAISTATS{#1}, #1}
  \newcommand{\pITCS}[1]{Preliminary\ version\ in\ the\ \cITCS{#1}, #1}
  \newcommand{\pPODC}[1]{Preliminary\ version\ in\ the\ \cPODC{#1}, #1}
  \newcommand{\pAPPROX}[1]{Preliminary\ version\ in\ the\ \cAPPROX{#1}, #1}
  \newcommand{\pSTACS}[1]{Preliminary\ version\ in\ the\ \cSTACS{#1}, #1}
  \newcommand{\pMTNS}[1]{Preliminary\ version\ in\ the\ \cMTNS{#1}, #1}
  \newcommand{\pICM}[1]{Preliminary\ version\ in\ the\ \cICM{#1}, #1}

  \newcommand{\STOC}[1]{Proceedings\ of\ the\ \cSTOC{#1}}
  \newcommand{\FSTTCS}[1]{Proceedings\ of\ the\ \cFSTTCS{#1}}
  \newcommand{\CCC}[1]{Proceedings\ of\ the\ \cCCC{#1}}
  \newcommand{\FOCS}[1]{Proceedings\ of\ the\ \cFOCS{#1}}
  \newcommand{\RANDOM}[1]{Proceedings\ of\ the\ \cRANDOM{#1}}
  \newcommand{\ISSAC}[1]{Proceedings\ of\ the\ \cISSAC{#1}}
  \newcommand{\ICALP}[1]{Proceedings\ of\ the\ \cICALP{#1}}
  \newcommand{\COLT}[1]{Proceedings\ of\ the\ \cCOLT{#1}}
  \newcommand{\CSR}[1]{Proceedings\ of\ the\ \cCSR{#1}}
  \newcommand{\MFCS}[1]{Proceedings\ of\ the\ \cMFCS{#1}}
  \newcommand{\PODS}[1]{Proceedings\ of\ the\ \cPODS{#1}}
  \newcommand{\SODA}[1]{Proceedings\ of\ the\ \cSODA{#1}}
  \newcommand{\NIPS}[1]{\cNIPS{#1}}
  \newcommand{\WALCOM}[1]{Proceedings\ of\ the\ \cWALCOM{#1}}
  \newcommand{\SoCG}[1]{Proceedings\ of\ the\ \cSoCG{#1}}
  \newcommand{\KDD}[1]{Proceedings\ of\ the\ \cKDD{#1}}
  \newcommand{\ICML}[1]{Proceedings\ of\ the\ \cICML{#1}}
  \newcommand{\AISTATS}[1]{Proceedings\ of\ the\ \cAISTATS{#1}}
  \newcommand{\ITCS}[1]{Proceedings\ of\ the\ \cITCS{#1}}
  \newcommand{\PODC}[1]{Proceedings\ of\ the\ \cPODC{#1}}
  \newcommand{\APPROX}[1]{Proceedings\ of\ the\ \cAPPROX{#1}}
  \newcommand{\STACS}[1]{Proceedings\ of\ the\ \cSTACS{#1}}
  \newcommand{\MTNS}[1]{Proceedings\ of\ the\ \cMTNS{#1}}
  \newcommand{\ICM}[1]{Proceedings\ of\ the\ \cICM{#1}}

  \newcommand{\arXiv}[1]{\href{http://arxiv.org/abs/#1}{arXiv:#1}}
  \newcommand{\farXiv}[1]{Full\ version\ at\ \arXiv{#1}}
  \newcommand{\parXiv}[1]{Preliminary\ version\ at\ \arXiv{#1}}
  \newcommand{\CoRR}{Computing\ Research\ Repository\ (CoRR)}

  \newcommand{\cECCC}[2]{\href{http://eccc.hpi-web.de/report/20#1/#2/}{Electronic\ Colloquium\ on\ Computational\ Complexity\ (ECCC),\ Technical\ Report\ TR#1-#2}}
  \newcommand{\ECCC}{Electronic\ Colloquium\ on\ Computational\ Complexity\ (ECCC)}
  \newcommand{\fECCC}[2]{Full\ version\ in\ the\ \cECCC{#1}{#2}}
  \newcommand{\pECCC}[2]{Preliminary\ version\ in\ the\ \cECCC{#1}{#2}}

\usepackage{url}

\begin{document}

	\begin{titlepage}
		\def\thepage{}
		\thispagestyle{empty}
		
		\title{A Direct $\tilde{O}(1/\epsilon)$ Iteration Parallel Algorithm for Optimal Transport} 
		
		\date{}
		\author{
			Arun Jambulapati\thanks{This material is based on work supported by NSF Graduate Fellowship DGE-114747.}\\
			Stanford University \\
			{\tt jmblpati@stanford.edu} 
			\and
			Aaron Sidford\thanks{This material is based on work supported by NSF CAREER Award CCF-1844855.}\\
			Stanford University \\
			{\tt sidford@stanford.edu}
			\and
			Kevin Tian\thanks{This material is based on work supported by NSF Graduate Fellowship DGE-1656518.}\\
			Stanford University \\
			{\tt kjtian@stanford.edu}
		}
		
		\maketitle
		
		\abstract{
Optimal transportation, or computing the Wasserstein or ``earth mover's'' distance between two $n$-dimensional distributions, is a fundamental primitive which arises in many learning and statistical settings. We give an algorithm which solves the problem to additive $\epsilon$ accuracy with $\tilde{O}(1/\epsilon)$ parallel depth and $\tilde{O}\left(n^2/\epsilon\right)$ work. 
 \cite{BlanchetJKS18, Quanrud19} obtained this runtime through reductions to positive linear programming and matrix scaling. However, these reduction-based algorithms use subroutines which may be impractical due to requiring solvers for second-order iterations (matrix scaling) or non-parallelizability (positive LP). Our methods match the previous-best work bounds by \cite{BlanchetJKS18, Quanrud19} while either improving parallelization or removing the need for linear system solves, and improve upon the previous best first-order methods running in time $\tilde{O}(\min(n^2 / \epsilon^2, n^{2.5}  / \epsilon))$ \cite{DvurechenskyGK18, LinHJ19}. We obtain our results by a primal-dual extragradient method, motivated by recent theoretical improvements to maximum flow \cite{Sherman17}.
}

	\end{titlepage}

	\section{Introduction}
\label{sec:intro}

Optimal transport is playing an increasingly important role as a subroutine in tasks arising in machine learning \cite{ArjovskyCB17}, computer vision \cite{BonneelPPH11, SolomonGPCBNDG15}, robust optimization \cite{EsfahaniK18, BlanchetK17}, and statistics \cite{PanaretosZ16}. Given these applications for large scale learning, designing algorithms for efficiently approximately solving the problem has been the subject of extensive recent research \cite{Cuturi13, AltschulerWR17, GenevayCPB16, ChakrabartyK18, DvurechenskyGK18, LinHJ19, BlanchetJKS18, Quanrud19}.

 Given two vectors $r$ and $c$ in the $n$-dimensional probability simplex $\Delta^n$ and a cost matrix $C \in \R^{n \times n}_{\geq 0}$\footnote{Similarly to earlier works, we focus on square matrices; generalizations to rectangular matrices are straightforward.}, the optimal transportation problem is
\begin{equation}
\label{eq:ot}
\min_{X \in \urc} \inprod{C}{X}, \quad \text{where} \quad \urc \defeq \left\{X \in \mathbb{R}^{n \times n}_{\geq 0},\; X\1 = r,\; X^\top\1 = c\right\}.
\end{equation}
This problem arises from defining the \emph{Wasserstein} or \emph{Earth mover's} distance between discrete probability measures $r$ and $c$, as the cheapest coupling between the distributions, where the cost of the coupling $X \in \urc$ is $\inprod{C}{X}$. If $r$ and $c$ are viewed as distributions of masses placed on $n$ points in some space (typically metric), the Wasserstein distance is the cheapest way to move mass to transform $r$ into $c$. In \eqref{eq:ot}, $X$ represents the transport plan ($X_{ij}$ is the amount moved from $r_i$ to $c_j$) and $C$ represents the cost of movement ($C_{ij}$ is the cost of moving mass from $r_i$ to $c_j$).

Throughout, the value of \eqref{eq:ot} is denoted $\opt$. We call $\hat{X}  \in \mathcal{U}_{r, c}$ an \emph{$\epsilon$-approximate transportation plan} if $\inprod{C}{\hat{X}} \leq \opt + \epsilon$. Our goal is to design an efficient algorithm to produce such a $\hat{X}$.

\subsection{Our Contributions}

Our main contribution is an algorithm running in $\tilde{O}(\cmax/\epsilon)$ parallelelizable iterations\footnote{Our iterations consist of vector operations and matrix-vector products, which are easily parallelizable. Throughout $\cmax$ is the largest entry of $C$.} and $\tilde{O}(n^2\cmax/\epsilon)$ total work producing an $\epsilon$-approximate transport plan.

Matching runtimes were given in the recent work of \cite{BlanchetJKS18, Quanrud19}. Their runtimes were obtained via reductions to matrix scaling and positive linear programming, each well-studied problems in theoretical computer science. However, the matrix scaling algorithm is a second-order Newton-type method which makes calls to structured linear system solvers, and the positive LP algorithm is not parallelizable (i.e. has depth polynomial in dimension). These features potentially limit the practicality of these 
algorithms. The key remaining open question this paper addresses is, \emph{is there an efficient first-order, parallelizable algorithm for approximating optimal transport?} We answer this affirmatively and give an efficient, parallelizable primal-dual first-order method; the only additional overhead is a scheme for implementing steps, incurring roughly an additional $\log \epsilon^{-1}$ factor. 

Our approach heavily leverages the recent improvement to the maximum flow problem, and more broadly two-player games on a simplex ($\ell_1$ ball) and a box ($\ell_\infty$ ball), due to the breakthrough work of \cite{Sherman17}. First, we recast \eqref{eq:ot} as a minimax game between a box and a simplex, proving correctness via a rounding procedure known in the optimal transport literature. Second, we show how to adapt the dual extrapolation scheme under the weaker convergence requirements of area-convexity, following \cite{Sherman17}, to obtain an approximate minimizer to our primal-dual objective in the stated runtime. En route, we slightly simplify analysis in \cite{Sherman17} and relate it more closely to the existing extragradient literature.

Finally, we give preliminary experimental evidence showing our algorithm can be practical, and highlight some open directions in bridging the gap between theory and practice of our method, as well as accelerated gradient schemes \cite{DvurechenskyGK18, LinHJ19} and Sinkhorn iteration.

\subsection{Previous Work}

{\bf Optimal Transport.} The problem of giving efficient algorithms to find $\epsilon$-approximate transport plans $\hat{X}$  which run in nearly linear time\footnote{We use ``nearly linear'' to describe complexities which have an $n^2 \textrm{polylog}(n)$ dependence on the dimension (where the size of input $C$ is $n^2$), and polynomial dependence on $\norm{C}_{\max}, \epsilon^{-1}$.} has been addressed by a line of recent work, starting with \cite{Cuturi13} and improved upon in \cite{GenevayCPB16, AltschulerWR17, DvurechenskyGK18, LinHJ19, BlanchetJKS18, Quanrud19}. We briefly discuss their approaches here.

Works by \cite{Cuturi13, AltschulerWR17} studied the Sinkhorn algorithm, an alternating minimization scheme. Regularizing \eqref{eq:ot} with an $\eta^{-1}$ multiple of entropy and computing the dual, we arrive at the problem
\begin{equation*}
\label{eq:sink_dual}
\min_{x,y \in \mathbb{R}^n} \textbf{1}^\top B_{\eta C}(x,y) \textbf{1} - r^\top x - c^\top y \quad \text{where} \quad B_{\eta C}(x,y)_{ij} = e^{x_i+y_j- \eta C_{ij}}.
\end{equation*}
This problem is equivalent to computing diagonal scalings $X$ and $Y$ for $M=\exp(-\eta C)$ such that $XMY$ has row sums $r$ and column sums $c$. The Sinkhorn iteration alternates fixing the row sums and the column sums by left and right scaling by diagonal matrices until an approximation of such scalings is found, or equivalently until $XMY$ is close to being in $\urc$.

As shown in \cite{AltschulerWR17}, we can round the resulting almost-transportation plan to a transportation plan which lies in $\urc$ in linear time, losing at most $2\cmax (\norm{X\1 - r}_1 + \norm{X^\top\1 - c}_1)$ in the objective. Further, \cite{AltschulerWR17} showed that $\tilde{O}(\cmax^3/\epsilon^3)$ iterations of this scheme sufficed to obtain a matrix which $\epsilon/\cmax$-approximately meets the demands in $\ell_1$ with good objective value, by analyzing it as an instance of mirror descent with an entropic regularizer. The same work proposed an alternative algorithm, Greenkhorn, based on greedy coordinate descent. \cite{DvurechenskyGK18, LinHJ19}  showed that $\tildeOh{\cmax^2/\epsilon^2}$ iterations, corresponding to $\tildeOh{n^2\cmax^2/\epsilon^2}$ work, suffice for both Sinkhorn and Greenkhorn, the current state-of-the-art for this line of analysis.

An alternative approach based on first-order methods was studied by \cite{DvurechenskyGK18, LinHJ19}. These works considered minimizing an entropy-regularized Equation~\ref{eq:ot}; the resulting weighted softmax function is prevalent in the literature on approximate linear programming \cite{Nesterov05}, and has found similar applications in near-linear algorithms for maximum flow \cite{Sherman13, KelnerLOS14, SidfordT18} and positive linear programming \cite{Young01, ZhuO15}. An unaccelerated algorithm, viewable as $\ell_\infty$ gradient descent, was analyzed in \cite{DvurechenskyGK18} and ran in $\tilde{O}(\cmax/\epsilon^2)$ iterations. Further, an accelerated algorithm was discussed, for which the authors claimed an $\tilde{O}(n^{1/4} \cmax^{0.5} / \epsilon)$ iteration count. \cite{LinHJ19} showed that the algorithm had an additional dependence on a parameter as bad as $n^{1/4}$, roughly due to a gap between the $\ell_2$ and $\ell_\infty$ norms. Thus, the state of the art runtime in this line is the better of $\tildeOh{n^{2.5}\cmax^{0.5}/\epsilon}$, $\tildeOh{n^2 \cmax/\epsilon^2}$ operations. The dependence on dimension of the former of these runtimes matches that of the linear programming solver of \cite{LeeS14, LeeS15}, which obtain a polylogarithmic dependence on $\epsilon^{-1}$, rather than a polynomial dependence; thus, the question of obtaining an accelerated $\epsilon^{-1}$ dependence without worse dimension dependence remained open.

This was partially settled in \cite{BlanchetJKS18, Quanrud19}, which studied the relationship of optimal transport to fundamental algorithmic problems in theoretical computer science, namely positive linear programming and matrix scaling, for which significantly-improved runtimes have been recently obtained \cite{ZhuO15, Allen-ZhuLOW17, CohenMTV17}. In particular, they showed that optimal transport could be reduced to instances of either of these objectives, for which $\tildeOh{\cmax/\epsilon}$ iterations, each of which required linear $O(n^2)$ work, sufficed. However, both of these reductions are based on black-box methods for which practical implementations are not known; furthermore, in the case of positive linear programming a parallel $\tilde{O}(1/\epsilon)$-iteration algorithm is not known. \cite{BlanchetJKS18} also showed any polynomial improvement to the runtime of our paper in the dependence on either $\epsilon$ or $n$ would result in maximum-cardinality bipartite matching in dense graphs faster than $\tilde{O}(n^{2.5})$ without fast matrix multiplication \cite{Sankowski09}, a fundamental open problem unresolved for almost 50 years \cite{HopcroftK73}.

\begin{table}
	\begin{center}
		\begin{tabular}{|c|c|c|c|c|c|} 
			\hline
			\textbf{Year} & \textbf{Author} & \textbf{Complexity} & \textbf{Approach} & \textbf{1st-order} & \textbf{Parallel} \\ [0.5ex] 
			\hline\hline
			2015 & \cite{LeeS15} & $\tilde{O}(n^{2.5})$ & Interior point & No & No \\
			\hline
			2017-19 & \cite{AltschulerWR17} & $\tilde{O}(n^2 \cmax^2/\epsilon^2)$ & Sink/Greenkhorn & Yes & Yes\\
			\hline
			2018 & \cite{DvurechenskyGK18} & $\tilde{O}(n^2 \cmax/\epsilon^2)$ & Gradient descent & Yes & Yes \\
			\hline
			2018-19 & \cite{LinHJ19} & $\tilde{O}(n^{2.5}\cmax^{0.5}/\epsilon)$ & Acceleration & Yes & Yes \\
			\hline
			2018 & \cite{BlanchetJKS18} & $\tilde{O}(n^2 \cmax / \epsilon)$ & Matrix scaling & No & Yes\\
			\hline
			2018-19 & \cite{BlanchetJKS18, Quanrud19} & $\tilde{O}(n^2 \cmax / \epsilon)$ & Positive LP & Yes & No \\
			\hline
			2019 & This work & $\tilde{O}(n^2 \cmax / \epsilon)$ & Dual extrapolation & Yes & Yes \\
			\hline
		\end{tabular}
	\end{center}
	\caption{Optimal transport algorithms. Algorithms using second-order information use potentially-expensive SDD system solvers; the runtime analysis of Sink/Greenkhorn is due to \cite{DvurechenskyGK18, LinHJ19}.} 
\end{table}

Specializations of the transportation problem to $\ell_p$ metric spaces or arising from geometric settings have been studied \cite{SharathkumarA12, AgarwalS14, AndoniNOY14}. These specialized approaches seem fundamentally different than those concerning the more general transportation problem.

Finally, we note recent work \cite{AltschulerBRW18} showed the promise of using the Nyström method for low-rank approximations to achieve speedup in theory and practice for transport problems arising from specific metrics. We find it interesting to combine our method with these improvements, and believe that as our method is based on matrix-vector operations, it is amenable to similar speedups.

{\it Remark.} During the revision process for this work, an independent result \cite{LahnMR19} was published to arXiv, obtaining improved runtimes for optimal transport via a combinatorial algorithm. The work obtains a runtime of $\tilde{O}(n^2\cmax/\epsilon + n\cmax^2/\epsilon^2)$, which is worse than our runtime by a low-order term. Furthermore, it does not appear to be parallelizable.

{\bf Box-simplex objectives.} Our main result follows from improved algorithms for bilinear minimax problems over one simplex domain and one box domain developed in \cite{Sherman17}. This fundamental minimax problem captures $\ell_1$ and $\ell_\infty$ regression over a simplex and box respectively, and inspired the development of conjugate smoothing \cite{Nesterov05} as well as mirror prox / dual extrapolation \cite{Nemirovski04, Nesterov07}. These latter two approaches are extragradient methods (using two gradient operations per iteration rather than one) for approximately solving a family of problems, which includes convex minimization and finding a saddle point to a convex-concave function. These methods simulate backwards Euler discretization of the gradient flow, similar to how mirror descent simulates forwards Euler discretization \cite{DiakonikolasO19}. The role of the extragradient step is a fixed point iteration (of two steps) which is a good approximation of the backwards Euler step when the operator is Lipschitz.

Nonetheless, the analysis of \cite{Nemirovski04, Nesterov07} fell short in obtaining a $1/T$ rate of convergence without worse dependence on dimension for these domains, where $T$ is the iteration count (which would correspond to a $\tildeOh{1/\epsilon}$ runtime for approximate minimization). The fundamental barrier was that over a box, any strongly-convex regularizer in the $\ell_\infty$ norm has a dimension-dependent domain size (shown in \cite{SidfordT18}). This barrier can also be viewed as the reason for the worse dimension dependence in the accelerated scheme of \cite{DvurechenskyGK18, LinHJ19}. 

The primary insight of \cite{Sherman17} was that previous approaches attempted to regularize the schemes of \cite{Nemirovski04, Nesterov07} with separable regularizers, i.e. the sum of a regularizer which depends only on the primal block and one which depends only on the dual. If, say, the domain of the primal block was a box, then such a regularization scheme would run into the $\ell_\infty$ barrier and incur a worse dependence on dimension. However, by more carefully analyzing the requirements of these algorithms, \cite{Sherman17} constructed a non-separable regularizer with small domain size, satisfying a property termed \emph{area-convexity} which sufficed for provable convergence of dual extrapolation \cite{Nesterov07}. Interestingly, the property seems specialized to dual extrapolation and not mirror prox \cite{Nemirovski04}.
	\section{Overview}
\label{sec:overview}

First, in Section~\ref{sec:reformulation} we first describe a reformulation of \eqref{eq:ot} as a primal-dual objective, which we solve approximately in Section~\ref{sec:extragradient}. Then in Section~\ref{sec:notation} we give additional notation critical for our analysis. In Section~\ref{sec:extragradient} we leverage this to give an overview of our main algorithm.

\subsection{$\ell_1$-regression formulation}
\label{sec:reformulation}

We adapt the view of \cite{BlanchetJKS18, Quanrud19} of the objective \eqref{eq:ot} as a positive linear program. Let $d$ be the (vectorized) cost matrix $C$ associated with the instance and let $\Delta^{n^2}$ be the $n^2$ dimensional simplex\footnote{We use $d$ because $C$ often arises from distances in a metric space, and to avoid overloading $c$.}. We recall $r, c$ are specified row and column sums with $\1^\top r = \1^\top c = 1$. The optimal transport problem can be written as, for $m = n^2$, and $A \in \{0, 1\}^{2n \times m}, b \in \R_{\geq 0}^{2n}$, for $A$ the (unsigned) \emph{edge-incidence matrix} of the underlying bipartite graph and $b$ the concatenation of $r$ and $c$.
\begin{equation}
\label{eq:otobj}
\min_{x \in \Delta^n, Ax = b} d^\top x.
\end{equation}
\begin{figure}[ht]
	\begin{equation*}
	A = \begin{pmatrix}
	1 & 1 & 1 & 0 & 0 & 0 & 0 & 0 & 0 \\
	0 & 0 & 0 & 1 & 1 & 1 & 0 & 0 & 0 \\
	0 & 0 & 0 & 0 & 0 & 0 & 1 & 1 & 1 \\
    1 & 0 & 0 & 1 & 0 & 0 & 1 & 0 & 0 \\
	0 & 1 & 0 & 0 & 1 & 0 & 0 & 1 & 0 \\
	0 & 0 & 1 & 0 & 0 & 1 & 0 & 0 & 1 
	\end{pmatrix},\;
	b = \begin{pmatrix}
	1/3 \\
	1/3 \\
	1/3 \\
	1/3 \\
	1/3 \\
	1/3 
	\end{pmatrix}.
	\end{equation*}
	\caption{Edge-incidence matrix $A$ of a $3 \times 3$ bipartite graph and uniform demands.}	
\end{figure}

In particular, $A$ is the 0-1 matrix on $V \times E$ such that $A_{ve} = 1$ iff $v$ is an endpoint of edge $e$. We summarize some additional properties of the constraint matrix $A$ and vector $b$.

\begin{fact} $A$, $b$ have the following properties.
	\begin{enumerate}
		\item $A \in \{0, 1\}^{2n \times m}$ has 2-sparse columns and $n$-sparse rows. Thus $\norm{A}_{1 \rightarrow 1} = 2$.
		\item $b^\top = \begin{pmatrix} r^\top & c^\top \end{pmatrix}$, so that $\norm{b}_1 = 2$.
		\item $A$ has $n^2$ nonzero entries.
	\end{enumerate}
\end{fact}

Section~\ref{sec:rounding} recalls the proof of the following theorem, which first appeared in \cite{AltschulerWR17}.

\begin{restatable}[Rounding guarantee, Lemma 7 in \cite{AltschulerWR17}]{theorem}{restateRounding}
\label{thm:rounding}
There is an algorithm which takes $\tilde{x}$ with $\norm{A\tilde{x} - b}_1 \leq \delta$ and produces $\hat{x}$ in $O(n^2)$ time, with
\begin{equation*}
A\hat{x} = b, \norm{\tilde{x} - \hat{x}}_1 \leq 2\delta.
\end{equation*}
\end{restatable}

We now show how the rounding procedure gives a roadmap for our approach. Consider the following $\ell_1$ regression objective over the simplex (a similar penalized objective appeared in \cite{Sherman13}): 
\begin{equation}
\label{eq:l1regress}
\min_{x \in \Delta^m} d^\top x + 2\dmax \norm{Ax - b}_1.
\end{equation}
We show that the penalized objective value is still $\opt$, and furthermore any approximate minimizer yields an approximate transport plan.

\begin{restatable}[Penalized $\ell_1$ regression]{lemma}{restateLRounding}
\label{lem:l1rounding}
The value of \eqref{eq:l1regress} is $\opt$. Also, given $\tilde{x}$, an $\epsilon$-approximate minimizer to \eqref{eq:l1regress}, we can find $\epsilon$-approximate transportation plan $\hat{x}$ 
in $O(n^2)$ time.
\end{restatable}
\begin{proof}
	Recall $\opt = \min_{Ax = b} d^\top x$. Let $\tilde{x}$ be the minimizing argument in \eqref{eq:l1regress}. We claim there is some optimal $\tilde{x}$ with $A\tilde{x} = b$; clearly, the first claim is then true. Suppose otherwise, and let $\norm{Ax - b}_1 = \delta > 0$. Then, let $\hat{x}$ be the result of the algorithm in Theorem~\ref{thm:rounding}, applied to $\tilde{x}$, so that $A\hat{x} = b, \norm{\tilde{x} - \hat{x}}_1 \leq 2\delta$. We then have
	\begin{equation*}
	d^\top \hat{x} + 2\dmax\norm{A\hat{x} - b}_1 = d^\top(\hat{x} - \tilde{x}) + d^\top \tilde{x} \leq d^\top \tilde{x} + \dmax \norm{\hat{x} - \tilde{x}}_1 \leq d^\top \tilde{x} + 2\dmax\delta.
	\end{equation*}
	The objective value of $\hat{x}$ is no more than of $\tilde{x}$, a contradiction. By this discussion, we can take any approximate minimizer to \eqref{eq:l1regress} and round it to a transport plan without increasing the objective.
\end{proof}

Section~\ref{sec:extragradient} proves Theorem~\ref{thm:mainthm}, which says we can efficiently find an approximate minimizer to \eqref{eq:l1regress}. 

\begin{theorem}[Approximate $\ell_1$ regression over the simplex]
\label{thm:mainthm}
There is an algorithm (Algorithm~\ref{alg:dualex}) taking  input $\epsilon$, which has $O((\dmax\log n \log \gamma)/\epsilon)$ parallel depth for $\gamma = \log n \cdot \dmax / \epsilon$, and total work $O(n^2 (\dmax\log n \log \gamma)/\epsilon)$, and obtains $\tilde{x}$ an $\epsilon$-additive approximation to the objective in \eqref{eq:l1regress}.
\end{theorem} 

We will approach proving Theorem~\ref{thm:mainthm} through a primal-dual viewpoint, in light of the following (based on the definition of the $\ell_1$ norm):
\begin{equation}
\label{eq:pdl1reg}
\min_{x \in \Delta^m} d^\top x + 2\dmax\norm{Ax - b}_1 = \min_{x \in \Delta^m} \max_{y \in [-1, 1]^{2n}} d^\top x + 2\dmax\left(y^\top Ax - b^\top y\right).
\end{equation}

Further, a low-\emph{duality gap} pair to \eqref{eq:pdl1reg} yields an approximate minimizer to \eqref{eq:l1regress}.

\begin{lemma}[Duality gap to error]
Suppose $x, y$ is feasible ($x \in \Delta^m, y \in [-1, 1]^{2n}$), and for any feasible $u, v$,
\begin{equation*}
\left(d^\top x + 2\dmax\left(v^\top A x - b^\top v\right)\right) - \left(d^\top u + 2\dmax\left(y^\top A u - b^\top y\right)\right) \leq \delta.
\end{equation*}
Then, we have $d^\top x + 2\dmax\norm{Ax - b}_1 \leq \delta + \opt$.
\end{lemma}
\begin{proof}
The result follows from maximizing over $v$, and noting that for the minimizing $u$,
\begin{equation*}
d^\top u + 2\dmax\left(y^\top Au - b^\top y\right) \leq d^\top u + 2\dmax\norm{Au - b}_1 = \opt.
\end{equation*}
\end{proof}

Correspondingly, Section~\ref{sec:extragradient} gives an algorithm which obtains $(x, y)$ with bounded duality gap within the runtime of Theorem~\ref{thm:mainthm}.

\subsection{Notation}
\label{sec:notation}

$\R_{\geq 0}$ is the nonnegative reals. $\1$ is the all-ones vector of appropriate dimension when clear. The probability simplex is $\Delta^d \defeq \{ v \mid  v \in \R_{\geq 0}^d, \1^\top v = 1\} $. We say matrix $X$ is in the simplex of appropriate dimensions when its (nonnegative) entries sum to one. 

$\norm{\cdot}_1$ and $\norm{\cdot}_{\infty}$ are the $\ell_1$ and $\ell_\infty$ norms, i.e. $\norm{v}_1 = \sum_i |v_i|$ and $\norm{v}_\infty = \max_i |v_i|$. When $A$ is a matrix, we let $\norm{A}_{p \rightarrow q}$ be the matrix operator norm, i.e. $\textup{sup}_{\norm{v}_p = 1} \norm{Av}_q$, where $\norm{\cdot}_p$ is the $\ell_p$ norm. In particular, $\norm{A}_{1 \rightarrow 1}$ is the largest $\ell_1$ norm of a column of $A$.

Throughout $\log$ is the natural logarithm. For $x \in \Delta^d$, $h(x) = \sum_{i \in [d]} x_i \log x_i$ is (negative) entropy where $0 \log 0 = 0$ by convention. It is well-known that $\max_{x \in \Delta^d} h(x) - \min_{x \in \Delta^d} h(x) = \log d$.

We also use the Bregman divergence of a regularizer and the proximal operator of a divergence.

\begin{definition}[Bregman divergence]
For (differentiable) regularizer $r$ and $z, w$ in its domain, the \emph{Bregman divergence} from $z$ to $w$ is
\begin{equation*}
V^r_z(w) \defeq r(w) - r(z) - \inprod{\nabla r(z)}{w - z}.
\end{equation*}
\end{definition}

When $r$ is convex, the divergence is nonnegative and convex in the argument ($w$ in the definition).

\begin{definition}[Proximal operator]
For (differentiable) regularizer $r$, $z$ in its domain, and $g$ in the dual space (when the domain is in $\R^d$, so is the dual space), we define the \emph{proximal operator} as
\begin{equation*}
\prox_z(g) \defeq \textup{argmin}_w \left\{ \inprod{g}{w} + V^r_z(w) \right\}.
\end{equation*}
\end{definition}

Several variables have specialized meaning throughout. All graphs considered will be on $2n$ vertices with $m$ edges, i.e. $m = n^2$. $A \in \R^{2n \times m}$ is the edge-incidence matrix. $d$ is the vectorized cost matrix $C$. $b$ is the constraint vector, concatenating row and column constraints $r$, $c$. In algorithms for solving \eqref{eq:pdl1reg}, $x$ and $y$ are primal (in a simplex) and dual (in a box) variables respectively. In Section~\ref{sec:extragradient}, we adopt the linear programming perspective where the decision variable $x \in \Delta^m$ is a vector. In Section~\ref{sec:rounding}, for convenience we take the perspective where $X$ is an unflattened $n \times n$ matrix. $\urc$ is the feasible polytope: when the domain is vectors, $\urc$ is $x \mid Ax = b$, and when it is matrices, $\urc$ is $X \mid X\1 = r, X^\top\1 = c$ (by flattening $X$ this is consistent).
	\section{Main Algorithm}
\label{sec:extragradient}

This section describes our algorithm for finding a primal-dual pair $(x, y)$ with a small duality gap, with respect to the objective in \eqref{eq:pdl1reg}, which we restate here for convenience:
\begin{equation*}
\min_{x \in \mathcal{X}} \max_{y \in \mathcal{Y}} d^\top x + 2\dmax\left(y^\top Ax - b^\top y\right), \; \mathcal{X} \defeq \Delta^m,\; \mathcal{Y} \defeq [-1, 1]^{2n}.
\tag{Restatement of \eqref{eq:pdl1reg} }
\end{equation*}
Our algorithm is a specialization of the algorithm in \cite{Sherman17}. One of our technical contributions in this regard is an analysis of the algorithm which more closely relates it to the analysis of dual extrapolation \cite{Nesterov07}, an algorithm for finding approximate saddle points with a more standard analysis. In Section~\ref{ssec:dexp}, we give the algorithmic framework and convergence analysis. In Section~\ref{ssec:altmin}, we provide analysis of an alternating minimization scheme for implementing steps of the procedure. The same procedure was used in \cite{Sherman17} which claimed without proof the linear convergence rate of the alternating minimization; we hope the analysis will make the method more broadly accessible to the optimization community. We defer many proofs to Appendix~\ref{app:extragradient}.

\subsection{Dual Extrapolation Framework}
\label{ssec:dexp}

For an objective $F(x, y)$ convex in $x$ and concave in $y$, the standard way to measure the duality gap is to define the \emph{gradient operator} $g(x, y) = (\nabla_x F(x, y), -\nabla_y F(x, y))$, and show that for $z = (x, y)$ and any $u$ on the product space, the \emph{regret}, $\inprod{g(z)}{z - u}$, is small. Correspondingly, we define
\begin{equation*}
g(x, y) \defeq \left(d + 2\dmax A^\top y, \; 2\dmax(b - Ax)\right).
\end{equation*}
The dual extrapolation framework \cite{Nesterov07} requires a regularizer on the product space. The algorithm is simple to state; it takes two ``mirror descent-like'' steps each iteration, maintaining a state $s_t$ in the dual space\footnote{In this regard, it is more similar to the ``dual averaging'' or ``lazy'' mirror descent setup \cite{Bubeck15}.}. A typical setup is a Lipschitz gradient operator and a regularizer which is the sum of canonical strongly-convex regularizers in the norms corresponding to the product space $\mathcal{X}, \mathcal{Y}$. However, recent works have shown that this setup can be greatly relaxed and still obtain similar rates of convergence. In particular, \cite{Sherman17} introduced the following definition.

\begin{definition}[Area-convexity]
Regularizer $r$ is $\kappa$-area-convex with respect to operator $g$ if for any points $a, b, c$ in its domain,
\begin{equation}
\label{eq:areaconvex}
\kappa\left(r(a) + r(b) + r(c) - 3r\left(\frac{a + b + c}{3}\right)\right) \geq \inprod{g(b) - g(a)}{b - c}.
\end{equation}
\end{definition}

Area-convexity is so named because $\inprod{g(b) - g(a)}{b - c}$ can be viewed as measuring the ``area'' of the triangle with vertices $a, b, c$ with respect to some Jacobian matrix. In the case of bilinear objectives, the left hand side in the definition of area-convexity is invariant to permuting $a, b, c$, whereas the sign of the right hand side can be flipped by interchanging $a, c$, so area-convexity implies convexity. However, it does not even imply the regularizer $r$ is strongly-convex, a typical assumption for the convergence of mirror descent methods.

We state the algorithm for time horizon $T$; the only difference from \cite{Nesterov07} is a factor of 2 in defining $s_{t + 1}$, i.e. adding a $1/2\kappa$ multiple rather than $1/\kappa$. We find it of interest to explore whether this change is necessary or specific to the analysis of \cite{Sherman17}.

\begin{algorithm}[H]
	\caption{$\bar{w} = \texttt{Dual-Extrapolation}(\kappa, r, g, T)$: Dual extrapolation with area-convex $r$.}
	\label{alg:dualex}
	\begin{algorithmic}
		\STATE Initialize $s_0 = 0$, let $\bar{z}$ be the minimizer of $r$.
		\FOR {$t < T$}
		\STATE $z_t \gets \prox^r_{\bar{z}}(s_t)$.
		\STATE $w_t \gets \prox^r_{\bar{z}}\left(s_t + \frac{1}{\kappa}g(z_t)\right)$.
		\STATE $s_{t + 1} \gets s_t + \frac{1}{2\kappa} g(w_t)$.
		\STATE $t \gets t + 1$.
		\ENDFOR
		\STATE $\textbf{return}$ $\bar{w} \defeq \frac{1}{T} \sum_{t \in [T]} w_t$.
	\end{algorithmic}
\end{algorithm}

\begin{restatable}[Dual extrapolation convergence]{lemma}{restateDualEx}
	\label{lem:dualex}
	Suppose $r$ is $\kappa$-area-convex with respect to $g$. Further, suppose for some $u$, $\Theta \geq r(u) - r(\bar{z})$. Then, the output $\bar{w}$ to Algorithm~\ref{alg:dualex} satisfies
	\begin{equation*}
	\inprod{g(\bar{w})}{\bar{w} - u} \leq \frac{2\kappa \Theta}{T}.
	\end{equation*}
\end{restatable}

In fact, by more carefully analyzing the requirements of dual extrapolation we have the following.

\begin{restatable}{corollary}{restateDualExEps}
\label{corr:dualex}
Suppose in Algorithm~\ref{alg:dualex}, the proximal steps are implemented with $\epsilon'$ additive error. Then, the upper bound of the regret in Lemma~\ref{lem:dualex} is $2\kappa\Theta/T + \epsilon'$.
\end{restatable}

We now state a useful second-order characterization of area-convexity involving a relationship between the Jacobian of $g$ and the Hessian of $r$, which was proved in \cite{Sherman17}.

\begin{theorem}[Second-order area-convexity, Theorem 1.6 in \cite{Sherman17}]
	\label{thm:socarea}
	For bilinear minimax objectives, i.e. whose associated operator $g$ has Jacobian
	\begin{equation*}
	J = \begin{pmatrix}
	0 & M^\top \\
	-M & 0
	\end{pmatrix},
	\end{equation*}
	and for twice-differentiable $r$, if for any $z$ in the domain,
	\begin{equation*}
	\begin{pmatrix}
	\kappa \nabla^2 r(z) & -J \\
	J & \kappa \nabla^2 r(z)
	\end{pmatrix}
	\succeq 0,
	\end{equation*}
	then $r$ is $3\kappa$-area-convex with respect to $g$.
	%
	%
\end{theorem}

Finally, we complete the outline of the algorithm by stating the specific regularizer we use, which first appeared in \cite{Sherman17}. We then prove its 3-area-convexity with respect to $g$ by using Theorem~\ref{thm:socarea}. 
\begin{equation}
\label{eq:shermanreg}
r(x, y) = 2\dmax\left(10 \sum_{j \in [n]} x_j \log x_j + x^\top A^\top (y^2)\right),
\end{equation}
where $(y^2)$ is entry-wise.

\begin{restatable}[Area-convexity of the Sherman regularizer]{lemma}{restateAreaConvexity}
	\label{lem:rsoc}
	For the Jacobian $J$ associated with the objective in \eqref{eq:pdl1reg} and the regularizer $r$ defined in \eqref{eq:shermanreg}, we have
	\begin{equation*}
	\begin{pmatrix}
	\nabla^2 r(z) & -J \\
	J & \nabla^2 r(z)
	\end{pmatrix}
	\succeq 0.  
	\end{equation*}
\end{restatable}

We now give the proof of Theorem~\ref{thm:mainthm}, requiring some claims in Appendix~\ref{ssec:altmin} for the complexity of Algorithm~\ref{alg:dualex}. In particular, Appendix~\ref{ssec:altmin} implies that although the minimizer to the proximal steps cannot be computed in closed form because of non-separability, a simple alternating scheme converges to an approximate-minimizer in near-constant time.

\begin{proof}[Proof of Theorem~\ref{thm:mainthm}]
The algorithm is Algorithm~\ref{alg:dualex}, using the regularizer $r$ in \eqref{eq:shermanreg}. Clearly, in the feasible region the range of the regularizer is at most $20\dmax\log n + 4\dmax$, where the former summand comes from the range of entropy and the latter $\norm{A^\top}_{\infty} = 2$. Thus, we may choose $\Theta = O(\dmax\log n)$ in Lemma~\ref{lem:dualex}, since $\inprod{\nabla r(\bar{z})}{\bar{z} - u} \leq 0 \Rightarrow V^r_{\bar{z}}(u) \leq r(u) - r(\bar{z})$ for all $u$. 

By Theorem~\ref{thm:socarea} and Lemma~\ref{lem:rsoc}, $r$ is 3-area-convex with respect to $g$. By Corollary~\ref{corr:dualex}, $T = 12\Theta/\epsilon$ iterations suffice, implementing each proximal step to $\epsilon/2$-additive accuracy. Finally, using Theorem~\ref{thm:altcomplexity} to bound this implementation runtime concludes the proof.
\end{proof}
	\section{Rounding to $\urc$}
\label{sec:rounding}

We state the rounding procedure in \cite{AltschulerWR17} for completeness here, which takes a transport plan $\tilde{X}$ close to $\urc$ and transforms it into a plan which exactly meets the constraints and is close to $\tilde{X}$ in $\ell_1$, and then prove its correctness in Appendix~\ref{app:rounding}. Throughout $r(X) \defeq X\1, c(X) \defeq X^\top \1$.

\begin{algorithm}[H]
	\caption{$\hat{X} = \texttt{Rounding}(\tilde{X}, r, c)$: Rounding to feasible polytope}
	\label{alg:round}
	\begin{algorithmic}
		\STATE $X' \gets \diag{\min\left(\frac{r}{r(\tilde{X})}, 1\right)} \tilde{X}$.
		\STATE $X'' \gets X' \diag{\min\left(\frac{c}{c(X')}, 1\right)}$.
		\STATE $e_r \gets r - \1^\top r(X''), e_c \gets c - \1^\top c(X''), E \gets \1^\top e_r$.
		\STATE $\hat{X} \gets X'' + \frac{1}{E} e_r e_c^\top$.
		\STATE $\textbf{return}$ $\hat{X}$.
	\end{algorithmic}
\end{algorithm}
	\section{Experiments}
\label{sec:experiments}

We show experiments illustrating the potential of our algorithm to be useful in practice, by considering its performance on computing optimal transport distances on the MNIST dataset and comparing against algorithms in the literature including APDAMD \cite{LinHJ19} and Sinkhorn iteration. All comparisons are based on the number of matrix-vector multiplications (rather than iterations, due to our algorithm's alternating subroutine), the main computational component of all algorithms considered.

\begin{figure}[ht!]
\centering
	\begin{subfigure}{.48\textwidth}
		\centering
		\includegraphics[width=\linewidth]{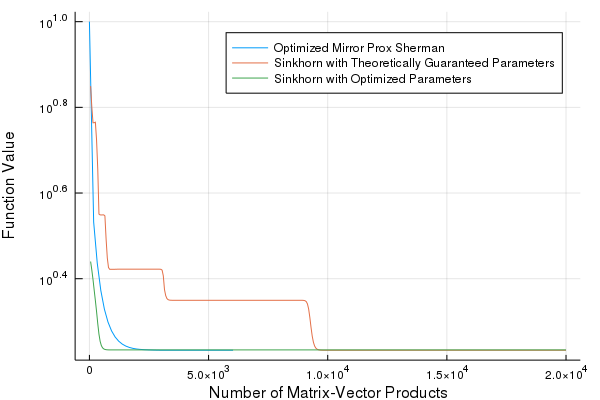}
		\caption{Comparison with Sinkhorn iteration.}
		\label{fig:sub1}
	\end{subfigure}
	\begin{subfigure}{.48\textwidth}
		\centering
		\includegraphics[width=\linewidth]{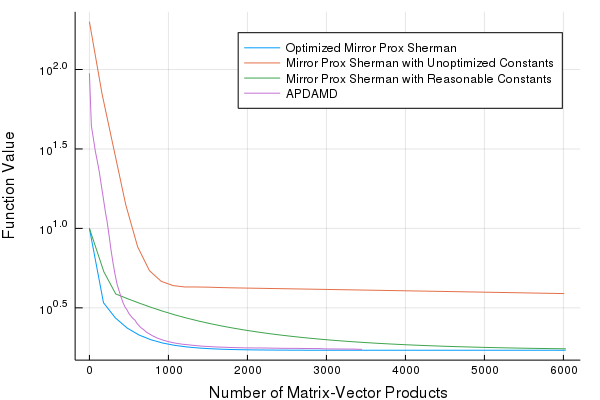}
		\caption{Comparison with APDAMD \cite{LinHJ19}.}
		\label{fig:sub2}
	\end{subfigure}
\label{fig:experiments}
\end{figure}

While our unoptimized algorithm performs poorly, slightly optimizing the size of the regularizer and step sizes used results in an algorithm with competitive performance to APDAMD, the first-order method with the best provable guarantees and observed practical performance. Sinkhorn iteration outperformed all first-order methods experimentally; however, an optimized version of our algorithm performed better than conservatively-regularized Sinkhorn iteration, and was more competitive with variants of Sinkhorn found in practice than other first-order methods.

As we discuss in our implementation details (Appendix~\ref{app:experiment_details}), we acknowledge that implementations of our algorithm illustrated are not the same as those with provable guarantees in our paper. However, we believe that our modifications are justifiable in theory, and consistent with those made in practice to existing algorithms. Further, we hope that studying the modifications we made (step size, using mirror prox \cite{Nemirovski04} for stability considerations), as well as the consideration of other numerical speedups such as greedy updates \cite{AltschulerWR17} or kernel approximations \cite{AltschulerBRW18}, will become fruitful for understanding the potential of accelerated first-order methods in both the theory and practice of computational optimal transport.

	\subsection*{Acknowledgments}
	We thank Jose Blanchet and Carson Kent for helpful conversations.
	
	
	\bibliographystyle{alpha}	
	\bibliography{optimal-transport}
	\newpage
	\begin{appendix}	
		
	\section{Algorithm}
\label{app:algorithm}

We give the complete algorithm for approximating optimal transport distance to additive $\epsilon$ here. We assume $C \in \R_{\geq 0}^{n \times n}$ and $r, c \in \Delta^n$. Finally, we refer to blocks of variable $z$ on a product space as $z^x, z^y$, i.e. $z = (z^x, z^y)$. Again $r(X) \defeq X\1$, $c(X) \defeq X^\top\1$.

\begin{algorithm}[H]
	\caption{$\hat{X} = \texttt{Optimal-Transport}(C, \epsilon, r, c)$: Produces $\epsilon$-approximate transportation plan}
	\label{alg:full}
	\begin{algorithmic}
		\STATE Vectorize $C$ to produce $d$.
		\STATE Let $b$ be $r, c$ concatenated; let $A$ be the incidence matrix of a complete $n \times n$ bipartite graph.
		\STATE $t \gets 0$.
		\STATE $x_0 \gets \frac{1}{n^2}\1$, $y_0 \gets \bf{0}_{2n}$.
		\STATE $s_0^x \gets \bf{0}_{n^2}$, $s_0^y \gets \bf{0}_{2n}$.
		\STATE $\Theta \gets 20\dmax \log n + 4\dmax$.
		\WHILE {$d^\top x_{t + \half} + 2\dmax\norm{Ax_{t + \half} - b}_1 \leq -2\dmax b^\top y_{t + \half} + \max_j \left[d + 2\dmax A^\top y_{t + \half}\right]_j + \epsilon$}
		\STATE $t \gets t + 1$.
		\STATE $k \gets 0$.
		\STATE $x'_0 \gets x_{t - \half}$, $y'_0 \gets y_{t - \half}$.
		\FOR{$0 \leq k < \ceil{24 \log\left(\left(\frac{88\dmax}{\epsilon^2} + \frac{2}{\epsilon}\right)\Theta\right)}$} 
		\STATE $x'_k \gets \exp\left(\frac{1}{20\dmax}s_t^x + \frac{1}{10} A^\top (y'_{k - 1})^2 \right)$, $x'_k \gets x'_k / \norm{x'_k}_1$.
		\STATE $y'_k \gets \min\left(1, \max\left(-1, \frac{-s_t^y}{4\dmax Ax'_k} \right)\right)$. Operations are element-wise.
		\ENDFOR
		\STATE $x_{t} \gets x'_k$, $y_{t} \gets y'_k$.
		\STATE $s_{t + \half}^x \gets s_t^x + \frac{1}{3}\left(d + 2\dmax A^\top y_{t}\right)$.
		\STATE $s^y_{t + \half} \gets s_t^y + \frac{1}{3}\left(2\dmax(b - Ax_{t})\right)$.
		\STATE $k \gets 0$.
		\STATE $x'_0 \gets x_{t}$, $y'_0 \gets y_{t}$. 
		\FOR {$0 \leq k < \ceil{24 \log\left(\left(\frac{88\dmax}{\epsilon^2} + \frac{2}{\epsilon}\right)\Theta\right)}$}
		\STATE $x'_k \gets \exp\left(\frac{1}{20\dmax}s_{t + \half}^x + \frac{1}{10} A^\top (y'_{k - 1})^2 \right)$, $x'_k \gets x'_k / \norm{x'_k}_1$.
		\STATE $y'_k \gets \min\left(1, \max\left(-1, \frac{-s_{t + \half}^y }{4\dmax Ax'_k} \right)\right)$. Operations are element-wise.
		\ENDFOR
		\STATE $x_{t + \half} \gets x'_k$, $y_{t + \half} \gets y'_k$.
		\STATE $s_{t + 1}^x \gets s_t^x + \frac{1}{6}\left(d + 2\dmax A^\top y_{t + \half}\right)$.
		\STATE $s^y_{t + 1} \gets s_t^y + \frac{1}{6}\left(2\dmax(b - Ax_{t + \half})\right)$.
		\ENDWHILE
		\STATE Un-vectorize $x$ to produce $\tilde{X}$.
		\STATE $X' \gets \diag{\min\left(\frac{r}{r(\tilde{X})}, 1\right)} \tilde{X}$.
		\STATE $X'' \gets X' \diag{\min\left(\frac{c}{c(X')}, 1\right)}$.
		\STATE $e_r \gets r - \1^\top r(X''), e_c \gets c - \1^\top c(X''), E \gets \1^\top e_r$.
		\STATE $\hat{X} \gets X'' + \frac{1}{E} e_r e_c^\top$.
		\STATE $\textbf{return}$ $\hat{X}$.
	\end{algorithmic}
\end{algorithm}

We remark that there are a variety of termination conditions that can be useful in practice for the alternating minimization procedure. For example, a standard early-stopping condition based on the observed movement of consecutive iterates was very successful in practice (Appendix~\ref{app:experiment_details}).
	\section{Missing proofs from Section~\ref{sec:extragradient}}
\label{app:extragradient}

In this section, we state missing proofs from Section~\ref{sec:extragradient}. We provide the efficient implementation of the proximal steps required by Algorithm~\ref{alg:dualex} in Appendix~\ref{ssec:altmin}.

\restateDualEx*
\begin{proof}
	Our first step is to prove the following inequality:
	\begin{equation}
	\label{eq:mainclaim}
	\frac{1}{2\kappa} \inprod{g(w_t)}{w_t - \bar{z}} \leq \inprod{s_{t + 1}}{z_{t + 1} - \bar{z}} + V^r_{\bar{z}}(z_{t + 1}) - \inprod{s_t}{z_t - \bar{z}} - V^r_{\bar{z}}(z_t).
	\end{equation}
	Let $c_t = \frac{z_t + w_t + z_{t + 1}}{3}$. The proof follows from minimality of $z_t$ with respect to $c_t$, minimality of $w_t$ with respect to $z_{t + 1}$, and area-convexity \eqref{eq:areaconvex} with respect to $z_t$, $w_t$, and $z_{t + 1}$. Respectively,
	\begin{equation}
	\label{eq:dexproof}
	\begin{aligned}
	\inprod{s_t}{z_t} + r(z_t) &\leq \inprod{s_t}{c_t} + r(c_t) \\
	\inprod{s_t}{w_t} + \frac{1}{\kappa}\inprod{g(z_t)}{w_t} + r(w_t) &\leq \inprod{s_t}{z_{t + 1}} + \frac{1}{\kappa}\inprod{g(z_t)}{z_{t + 1}} + r(z_{t + 1}) \\
	\frac{1}{\kappa}\inprod{g(w_t) - g(z_t)}{w_t - z_{t + 1}} &\leq r(z_t) + r(w_t) + r(z_{t + 1}) - 3r\left(c_t\right).
	\end{aligned}
	\end{equation}
	Substituting the first equation into the third and using the definition of $c_t$, we have
	\begin{equation*}
	\frac{1}{\kappa}\inprod{g(w_t) - g(z_t)}{w_t - z_{t + 1}} \leq r(w_t) + r(z_{t + 1}) - 2r(z_t) + \inprod{s_t}{w_t + z_{t + 1} - 2z_t}.
	\end{equation*}
	Rearranging the second equation, we have
	\begin{equation*}
	\frac{1}{\kappa}\inprod{g(z_t)}{w_t - z_{t + 1}} \leq r(z_{t + 1}) - r(w_t) + \inprod{s_t}{z_{t + 1} - w_t}.
	\end{equation*}
	Adding these two equations, we have
	\begin{equation*}
	\frac{1}{\kappa}\inprod{g(w_t)}{w_t - z_{t + 1}} \leq 2r(z_{t + 1}) - 2r(z_t) + \inprod{s_t}{2z_{t + 1} - 2z_t}.
	\end{equation*}
	Dividing by 2 and adding $\frac{1}{2\kappa}\inprod{g(w_t)}{z_{t + 1} - \bar{z}}$ to both sides, we obtain the desired \eqref{eq:mainclaim}. Now, define the potential function
	\begin{equation*}
	\Phi_k = \frac{1}{2\kappa}\sum_{t = 0}^{k - 1} \inprod{g(w_t)}{w_t - \bar{z}} - \inprod{s_k}{z_k - \bar{z}} - V^r_{\bar{z}}(z_{k})
	\end{equation*}
	Then, by \eqref{eq:mainclaim}, $\Phi_k$ is nonincreasing in $k$. Therefore for any $u$, by the definition of $\Theta$,
	\begin{align*}
	\frac{1}{T} \sum_{t = 0}^{T - 1}\inprod{g(w_t)}{w_t - u} &\leq \frac{1}{T} \sum_{t = 0}^{T - 1}\inprod{g(w_t)}{w_t - \bar{z}} + \frac{1}{T} \sum_{t = 0}^{T - 1}\inprod{g(w_t)}{\bar{z} - u} + \left(\frac{2\kappa\Theta}{T} - \frac{2\kappa V_{\bar{z}}(u)}{T}\right) \\
	&\leq \frac{1}{T} \sum_{t = 0}^{T - 1}\inprod{g(w_t)}{w_t - \bar{z}} + \frac{1}{T} \sum_{t = 0}^{T - 1}\inprod{g(w_t)}{\bar{z} - z_T} + \left(\frac{2\kappa\Theta}{T} - \frac{2\kappa V_{\bar{z}}(z_T)}{T}\right) \\
	&= \frac{2\kappa}{T} \Phi_T + \frac{2\kappa\Theta}{T} \leq \frac{2\kappa}{T} \Phi_0 + \frac{2\kappa\Theta}{T} = \frac{2\kappa\Theta}{T}.
	\end{align*}
	The inequality on the second line used the definition of $z_T = \prox^r_{\bar{z}}\left(\frac{1}{2\kappa}\sum_{t \in [T - 1]} g(w_t)\right)$, and the last inequality is  $\Phi_T \leq \Phi_0$. The conclusion follows from the definition of $g$ (because it is linear).
\end{proof}

\restateDualExEps*
\begin{proof}
	We see that \eqref{eq:mainclaim} now holds up to $\epsilon'$ additive error, so that $\Phi_k$ is increasing by at most $\epsilon'$ each step. Thus, we obtain $\Phi_T \leq \Phi_0 + T\epsilon'$, yielding the conclusion.
\end{proof}

\restateAreaConvexity*

\begin{proof}
	We scale both $r$ and $J$ down by $2\dmax$, which does not affect positive-semidefiniteness. By computation we have (recalling all columns of $A$ have $\ell_1$ norm of 2)
	\begin{equation*}
	\nabla^2 r(x, y) = \begin{pmatrix}
	5 \norm{A_{:j}}_1 \diag{\frac{1}{x_j}} & 2 A^\top \diag{y_i} \\
	2 \diag{y_i} A & 2 \diag{A_i^\top x} \\
	\end{pmatrix}.
	\end{equation*}
	It suffices to show that for any vector $\begin{pmatrix} a & b & c & d \end{pmatrix}$ we have
	\begin{equation*}
	\begin{pmatrix} a & b & c & d \end{pmatrix}
	\begin{pmatrix}
	5 \norm{A_{:j}}_1 \diag{\frac{1}{x_j}} & 2 A^\top \diag{y_i} & 0 & -A^\top \\
	2 \diag{y_i} A & 2 \diag{A_i^\top x} & A & 0 \\
	0 & A^\top & 5 \norm{A_{:j}}_1 \diag{\frac{1}{x_j}} & 2 A^\top \diag{y_i} \\
	-A & 0 & 2 \diag{y_i} A & 2 \diag{A_i^\top x}
	\end{pmatrix}
	\begin{pmatrix}
	a \\ b \\ c \\ d
	\end{pmatrix} \geq 0.
	\end{equation*}
	Upon simplifying and gathering like terms, it suffices to show
	\begin{equation*}
	\sum_{i, j} A_{ij}\left(\frac{5a_j^2}{x_j} + 4 a_j b_i y_i + 2 b_i^2 x_j - 2 a_j d_i + 2 c_j b_i + \frac{5 c_j^2}{x_j} + 4 c_j d_i y_i + 2 d_i^2 x_j \right) \geq 0.
	\end{equation*}
	However, this is true for $y_i \in [-1, 1]$, since each coefficient groups into clearly nonnegative terms,
	\begin{align*}
	\left(\frac{4a_j^2}{x_j} + 4 a_j b_i y_i + b_i^2 x_j \right) + \left(\frac{a_j^2}{x_j} - 2 a_j d_i + d_i^2 x_j\right) + \left(\frac{4 c_j^2}{x_j} + 4 c_j d_i y_i + d_i^2 x_j \right) + \left(\frac{c_j^2}{x_j} + 2 c_j b_i + b_i^2 x_j \right).
	\end{align*}
\end{proof}

\subsection{Alternating Minimization Analysis}
\label{ssec:altmin}

In this section, we give the convergence analysis of an alternating minimization procedure for minimizing a function of the form (throughout this section, $r(x, y)$ is as in \eqref{eq:shermanreg})
\begin{equation}
\label{eq:proxfunction}
f(x, y) \defeq \inprod{\xi}{x} + \inprod{\eta}{y} + r(x, y)
\end{equation}
which is the type of minimization problem arising from steps of the form $\prox^r_{\bar{z}}(g)$. As we will see, $f(x, y)$ is jointly convex. Throughout this section, let $x_{\opt}, y_{\opt}$ be the minimizer to $f$. Corollary~\ref{corr:dualex} states that $O(\epsilon)$ additive error to $f$ gives the same asymptotic convergence rate in Algorithm~\ref{alg:dualex}. We will show that a simple alternating minimization scheme enjoys a linear rate of convergence in our setting; thus, roughly $O(\log \epsilon^{-1})$ iterations suffice. We first give a proof of a general condition which suffices for linear convergence.

\begin{lemma}
	\label{lem:altminprog}
	Suppose $f(x, y)$ is twice-differentiable and jointly convex, over the product space $\mathcal{X} \times \mathcal{Y}$. Consider the alternating minimization scheme,
	\begin{enumerate}
		\item $x_{k + 1} \defeq \textup{argmin}_{x \in \mathcal{X}} f(x, y_k)$
		\item $y_{k + 1} \defeq \textup{argmin}_{y \in \mathcal{Y}} f(x_{k + 1}, y)$
	\end{enumerate}
	Further, suppose there are convex regions $\mathcal{X}_{k + 1} \subseteq \mathcal{X}$,  $\mathcal{Y}_{k} \subseteq \mathcal{Y}$ which contain $x_{k + 1}, y_k$ respectively, such that for any $x' \in \mathcal{X}_{k + 1}$, $y', y'' \in \mathcal{Y}_{k}$, and for some $\sigma \geq 1$,
	\begin{equation}
	\label{eq:hessdom}
	\nabla^2 f(x', y') \succeq \frac{1}{\sigma} \nabla^2_{yy} f(x_{k + 1}, y''),
	\end{equation}
	where $\nabla^2_{yy}$ is the Hessian with all but the $yy$ block zeroed out. Then, for any $x^* \in \mathcal{X}_{k + 1}$, $y^* \in \mathcal{Y}_k$,
	\begin{equation*}
	f(x_{k + 1}, y_k) - f(x_{k + 1}, y_{k + 1}) \geq \frac{1}{\sigma}\left( f(x_{k + 1}, y_k) - f(x^*, y^*)\right).
	\end{equation*}
\end{lemma}
\begin{proof}
	Let $\tilde{y} = \left(1 - \frac{1}{\sigma}\right) y_k + \frac{1}{\sigma} y^*$. We will prove instead that
	\begin{equation*}
	f(x_{k + 1}, y_k) - f(x_{k + 1}, \tilde{y}) \geq \frac{1}{\sigma}\left( f(x_{k + 1}, y_k) - f(x^*, y^*)\right),
	\end{equation*}
	from which the conclusion will follow since $f(x_{k + 1}, y_{k + 1}) \leq f(x_{k + 1}, \tilde{y})$. Note by definition of $\tilde{y}$, as well as optimality of $x_{k + 1}$ which implies $0 \geq \inprod{\nabla_x f(x_{k + 1}, y_k)}{x_{k + 1} - x^*}$,
	\begin{equation}
	\label{eq:firstorder}
	\inprod{\nabla_y f(x_{k + 1}, y_k)}{y_k - \tilde{y}} = \frac{1}{\sigma}\inprod{\nabla_y f(x_{k + 1}, y_k)}{y_k - y^*} \geq \frac{1}{\sigma} \inprod{\nabla f(x_{k + 1}, y_k)}{z_{k + \half} - z^*}
	\end{equation}
	where $z_{k + \half} \defeq (x_{k + 1}, y_k)$ and $z^* \defeq (x^*, y^*)$. Further, let $y_{\alpha} \defeq (1 - \alpha) y_k + \alpha y^*$, $\tilde{y}_{\alpha} \defeq (1 - \alpha) y_k + \alpha \tilde{y}$, and $x_{\alpha} \defeq (1 - \alpha) x_{k + 1} + \alpha x^*$. Then, by Taylor expansion we have $f(x_{k + 1}, y_k) - f(x_{k + 1}, \tilde{y})$ equals
	\begin{align*}
	&\inprod{\nabla_y f(x_{k + 1}, y_k)}{y_k - \tilde{y}} - \int_0^1 \int_0^\beta (\tilde{y} - y_k)^\top \nabla^2_{yy} f(x_{k + 1}, \tilde{y}_{\alpha}) (\tilde{y} - y_k) d\alpha d\beta \\
	\geq \;&\frac{1}{\sigma}\inprod{\nabla f(x_{k + 1}, y_k)}{z_{k + \half} - z^*} - \frac{1}{\sigma^2} \int_0^1 \int_0^\beta (y^* - y_k)^\top \nabla^2_{yy} f(x_{k + 1}, \tilde{y}_{\alpha}) (y^* - y_k) d\alpha d\beta \\
	\geq \;&\frac{1}{\sigma}\left(\inprod{\nabla f(x_{k + 1}, y_k)}{z_{k + \half} - z^*} - \int_0^1 \int_0^\beta (z^* - z_{k + \half})^\top \nabla^2 f(x_{\alpha}, y_{\alpha}) (z^* - z_{k + \half}) d\alpha d\beta\right) \\
	= \;&\frac{1}{\sigma}\left(f(x_{k + 1}, y_k) - f(x^*, y^*)\right).
	\end{align*}
	In the first inequality, we used \eqref{eq:firstorder} and the definition of $\tilde{y}$, and in the second we used \eqref{eq:hessdom} (since $x_{\alpha} \in \mathcal{X}_{k + 1}, y_{\alpha}, \tilde{y}_{\alpha} \in \mathcal{Y}_{k}$ by convexity).
\end{proof}

We now give a helper lemma specialized to the particular $f$ in \eqref{eq:proxfunction}, which will be used in the proof of convergence.

\begin{lemma}
	\label{lem:hessball}
	For some $x_{k + 1}, y_k$, let $\mathcal{X}_{k + 1} = \left\{x \mid x \geq \half x_{k + 1}\right\}$ where the inequality is entrywise, and let $\mathcal{Y}_k$ be the entire domain of $y$ (i.e. $\mathcal{Y}$). Then for any $x' \in \mathcal{X}_{k + 1}, y', y'' \in \mathcal{Y}_k$,
	\begin{equation*}
	\nabla^2 r(x', y') \succeq \frac{1}{12} \nabla^2_{yy} r(x_{k + 1}, y'').
	\end{equation*}
\end{lemma}
\begin{proof}
	Recall that (since $\norm{A_{:j}}_1 = 2$)
	\begin{equation*}
	\nabla^2 r(x, y) = 2\dmax \begin{pmatrix}
	5 \norm{A_{:j}}_1 \diag{\frac{1}{x_j}} & 2 A^\top \diag{y_i} \\
	2 \diag{y_i} A & 2 \diag{A_i^\top x} \\
	\end{pmatrix}.
	\end{equation*}
	Consider the diagonal approximation
	\begin{equation*}
	D(x) = 2\dmax \begin{pmatrix}
	\norm{A_{:j}}_1 \diag{\frac{1}{x_j}} & 0 \\
	0 & \diag{A_i^\top x} \\
	\end{pmatrix}.
	\end{equation*}
	We claim for any $y$, 
	\begin{equation}
	\label{eq:diagapprox}
	D(x) \preceq \nabla^2 r(x, y) \preceq 6D(x).
	\end{equation}
	To see this, consider the quadratic forms with respect to some vector $\begin{pmatrix} u & v \end{pmatrix}$:
	\begin{align*}
	\begin{pmatrix}
	u & v
	\end{pmatrix}
	\nabla^2 r(x, y)
	\begin{pmatrix}
	u \\
	v
	\end{pmatrix}
	&= 
	2\dmax \sum_{i, j} A_{ij} \left(\frac{5u_j^2}{x_j} + 4u_j v_i y_i + 2 v_i^2 x_j \right),\\
	\begin{pmatrix}
	u & v
	\end{pmatrix}
	D(x)
	\begin{pmatrix}
	u \\
	v
	\end{pmatrix}
	&= 
	2\dmax \sum_{i, j} A_{ij} \left(\frac{u_j^2}{x_j} + v_i^2 x_j \right).
	\end{align*}
	Now \eqref{eq:diagapprox} follows because for any $y_i \in [-1, 1]$, it's easy to verify
	\begin{equation*}
	\frac{u_j^2}{x_j} + v_i^2 x_j \leq \frac{5u_j^2}{x_j} + 4u_j v_i y_i + 2 v_i^2 x_j \leq 6\left(\frac{u_j^2}{x_j} + v_i^2 x_j\right).
	\end{equation*}
	Therefore, to prove the lemma statement we can use
	\begin{equation*}
	\nabla^2 r(x', y') \succeq D(x') \succeq \half D(x_{k + 1}) \succeq \frac{1}{12} \nabla^2_{yy} r(x_{k + 1}, y'').
	\end{equation*}
	The inequality $D(x') \succeq \half D(x_{k + 1})$ followed from the definition of $\mathcal{X}_{k + 1}$, and the last inequality followed from $D(x_{k + 1})$ spectrally dominating $\frac{1}{6}\nabla^2 r(x_{k + 1}, y'')$, and restrictions of $D(x_{k + 1})$ to the $yy$ block can only decrease the quadratic form.
\end{proof}

We now give the proof of the linear rate of convergence.

\begin{lemma}
	\label{lem:convergef}
	For $f(x, y)$ defined in \eqref{eq:proxfunction}, the alternating minimization scheme
	\begin{enumerate}
		\item $x_{k + 1} \defeq \textup{argmin}_{x \in \mathcal{X}} f(x, y_k)$.
		\item $y_{k + 1} \defeq \textup{argmin}_{y \in \mathcal{Y}} f(x_{k + 1}, y)$.
	\end{enumerate}
	decreases the function error $f(x_k, y_k) - f(x_{\opt}, y_{\opt})$ by a factor of at least $1/24$ in each iteration.
\end{lemma}
\begin{proof}
	We can apply Lemma~\ref{lem:altminprog} with the sets defined in Lemma~\ref{lem:hessball}, with $\sigma = 12$. On iteration $k$, consider picking the points $x^*, y^* = \half(x_{k + 1} + x_{\opt}), \half(y_k + y_{\opt})$. Evidently, $x^* \in \mathcal{X}_{k + 1}, y^* \in \mathcal{Y}_k$. Therefore, since $f(x_{k + 1}, y_{k + 1}) \geq f(x_{k + 2}, y_{k + 1})$,
	\begin{equation*}
	f(x_{k + 1}, y_k) - f(x_{k + 2}, y_{k + 1}) \geq f(x_{k + 1}, y_k) - f(x_{k + 1}, y_{k + 1}) \geq \frac{1}{12}(f(x_{k + 1}, y_k) - f(x^*, y^*)).
	\end{equation*}
	Furthermore, by convexity, we have
	\begin{equation*}
	f(x_{k + 1}, y_k) - f(x^*, y^*) \geq \half(f(x_{k + 1}, y_k) - f(x_{\opt}, y_{\opt})).
	\end{equation*}
	Finally, combining these two inequalities and rearranging,
	\begin{equation*}
	\frac{23}{24}(f(x_{k + 1}, y_k) - f(x_{\opt}, y_{\opt})) \geq f(x_{k + 2}, y_{k + 1}) - f(x_{\opt}, y_{\opt}).
	\end{equation*}
	Thus, by taking a $y$ step and then an $x$ step, we decrease the function error by a $1/24$ factor.
\end{proof}

Finally, we show that steps of the alternating minimization can be implemented in linear time.

\begin{lemma}
	\label{lem:fastimpl}
	For $f(x, y)$ defined in \eqref{eq:proxfunction}, we can implement the steps
	\begin{enumerate}
		\item $x_{k + 1} \defeq \textup{argmin}_x f(x, y_k)$.
		\item $y_{k + 1} \defeq \textup{argmin}_y f(x_{k + 1}, y)$.
	\end{enumerate}
	restricted to the relevant domains, in time $O(n^2)$.
\end{lemma}
\begin{proof}
	Recall $A$ has $n^2$ nonzero entries, so a matrix-vector multiplication can be performed in this time. Computing $x$ in linear time is straightforward: it is defined by
	\begin{equation*}
	\argmin_{x} \left\langle\gamma, x\right\rangle + \sum_{j \in [n]} x_j \log x_j \text{ such that } x \in \Delta^m, \gamma \defeq \frac{1}{20\dmax}\xi + \frac{1}{10} A^\top (y^2).
	\end{equation*}
	By examining the KKT conditions, it is clear that the minimizing $x$ is proportional to $\exp(-\gamma)$; computing $\gamma$ takes $O(n^2)$ time, as does the simplex projection. 	
	Similarly, computing $y$ in linear time is simple for fixed $x$: it is
	\begin{equation*}
	\argmin_{y} \inprod{\eta}{y} + \inprod{2\dmax Ax}{y^2} \text{ such that } y \in [-1, 1]^{2n},
	\end{equation*}
	which is coordinate-wise decomposable as minimizing a quadratic over an interval.
\end{proof}

\begin{theorem}[Complexity of alternating minimization]
\label{thm:altcomplexity}
We can obtain an $\epsilon/2$-approximate minimizer to the proximal steps required by Algorithm~\ref{alg:dualex} to $\epsilon/2$ accuracy, with the regularizer of \eqref{eq:shermanreg} and $\kappa = 3$, in $O(\log \gamma)$ parallelizable iterations for $\gamma = \log n \cdot \dmax \cdot \epsilon^{-1}$, and $O(n^2 \log \gamma)$ total work.
\end{theorem}
\begin{proof}
By Lemmas~\ref{lem:convergef} and~\ref{lem:fastimpl}, we can spend $O(n^2)$ parallelizable work to decrease the suboptimality gap by a $1/24$ factor, so it remains to argue that the initial error is at most $\textrm{poly}(\log n, \dmax, \epsilon^{-1})$ to show that implementing the proximal steps to additive error $\epsilon / 2$ can be done in $O(\log \gamma)$ iterations. We show that this is true for implementing the proximal step for $z_t$; a similar argument holds for $w_t$. To this end, note that by our setting of $\kappa$, for any $z$ where we let $g(z) = (g^x(z), g^y(z))$,
\begin{align*}
\frac{1}{2\kappa} \norm{g^x(z)}_{\infty} = \frac{1}{6}\norm{d + 2\dmax A^\top y}_{\infty} \leq \frac{\dmax}{2},\\
\frac{1}{2\kappa}\norm{g^y(z)}_1 = \frac{1}{6}\norm{2\dmax(b - Ax)}_1 \leq \frac{4\dmax}{3}.
\end{align*}
Therefore, for $s_t = (s^x_t, s^y_t)$, by the triangle inequality, and $t \leq 12\Theta / \epsilon$ the bound on the number of steps required where $\Theta$ is the range of $r$, we have
\begin{align*}
\norm{s_t^x}_{\infty} \leq t \cdot \frac{1}{2\kappa} \norm{g^x(z)}_{\infty} \leq \frac{6\dmax\Theta}{\epsilon},\\
\norm{s_t^y}_{1} \leq t \cdot \frac{1}{2\kappa}\norm{g^y(z)}_1 \leq \frac{16\dmax\Theta}{\epsilon}.
\end{align*}
A simple calculation yields $\Theta = 20\dmax \log n + 4\dmax$ upper bounds the range of $r$. Finally, let $x_t^*, y_t^*$ be the minimizer of the proximal objective,
\begin{equation*}
\inprod{s^x_t}{x} + \inprod{s^y_t}{y} + r(x, y).
\end{equation*}
For any initialization $x_{\textrm{init}}, y_{\textrm{init}}$ to the alternating minimization, the suboptimality gap is given by
\begin{align*}
\inprod{s^x_t}{x_{\textrm{init}} - x_t^*} + \inprod{s^y_t}{y_{\textrm{init}} - y_t^*} + r(x_{\textrm{init}}, y_{\textrm{init}}) - r(x_t^*, y_t^*)\\ 
\leq \norm{x_{\textrm{init}} - x_t^*}_1\norm{s^x_t}_\infty + \norm{y_{\textrm{init}} - y_t^*}_{\infty}\norm{s^y_t}_1 + \Theta \leq \left(\frac{44\dmax}{\epsilon} + 1 \right)\Theta. 
\end{align*}
Therefore, the total number of iterations required is bounded by $24 \log\left(\left(\frac{88\dmax}{\epsilon^2} + \frac{2}{\epsilon}\right)\Theta\right)$ as desired.
\end{proof}
	\section{Missing proofs from Section~\ref{sec:rounding}}
\label{app:rounding}

In this section, we give the proof to Theorem~\ref{thm:rounding}.

\restateRounding*

\begin{proof}
	The algorithm is Algorithm~\ref{alg:round}. We adopt the alternative view of $\tilde{x}$ as a $n \times n$ matrix $\tilde{X}$ in the simplex, and define operations $r(X) = X\1, c(X) = X^\top\1$, recalling the first and last $n$ entries of $b$ are $r, c$, i.e. the row and column constraints. Recall we assume we have
	\begin{equation*}
	\norm{r(\tilde{X}) - r}_1 + \norm{c(\tilde{X}) - c}_1 \leq \delta.
	\end{equation*}
	Clearly all operations in Algorithm~\ref{alg:round} take $O(n^2)$ time. To explain briefly, $X'$ is fixed so that its row sums are feasible (i.e. $X'\1 \leq r$) and $X''$ is fixed so that its column sums are feasible. Further, entrywise $X'' \leq X' \leq \tilde{X}$, so $X''$ is feasible. We first bound
	\begin{equation*}
	d \defeq \norm{X'' - \tilde{X}}_1 = \left(\sum_{i: r_i(\tilde{X}) > r_i} r_i(\tilde{X}) - r_i\right) + \left(\sum_{j: c_j(X') > c_j} c_j(X') - c_j\right).
	\end{equation*}
	Note $\norm{r(\tilde{X}) - r}_1 \geq \sum_{i: r_i(\tilde{X}) > r_i} r_i(\tilde{X}) - r_i$. Further, by $X' \leq \tilde{X}$ entrywise,
	\begin{equation*}
	\sum_{j: c_j(X') > c_j} c_j(X') - c_j \leq \norm{c(\tilde{X}) - c}_1.
	\end{equation*}
	Thus $d \leq \delta$. $\hat{X} \in \urc$, since $e_r, e_c \geq 0$ and $\1^\top e_r = \1^\top e_c = e$, so $\hat{X}\1 = r, \; \hat{X}^\top \1 = c$. Also,
	\begin{equation*}
	\norm{\hat{X} - \tilde{X}}_1 \leq \norm{X'' - \tilde{X}}_1 + \norm{\hat{X} - X''}_1 \leq \delta + e.
	\end{equation*}
	Finally,
	\begin{equation*}
	e = 1 - \1^\top X'' \1 = 1 - \left(\1^\top \tilde{X} \1 - d\right) = d.
	\end{equation*}
	Thus using $d \leq \delta$ proves the claim.
\end{proof}
	\section{Experiment details}
\label{app:experiment_details}

Here, we give the implementation details for the experimental results discussed in Section~\ref{sec:experiments}, and a brief justification of experimental decisions we made.

{\bf Dataset.} For both figures in Section~\ref{sec:experiments}, we had the following experimental setup. We randomly sampled a pair of digits from the MNIST dataset corresponding to the digit 1, and added a small amount of background noise for numerical stability, as is standard in the literature \cite{AltschulerWR17}. We downsampled the $28 \times 28$ pixel images to size $14 \times 14$ by skipping every other pixel to speed up experiments. Similar performances were observed across multiple random instances. Finally, the cost metric used was by Manhattan distance on the 2-dimensional grid.

{\bf Objective value.} For simplicity, in all cases we measured objective value by the overestimate presented in \eqref{eq:pdl1reg}. By the proof of Lemma~\ref{lem:l1rounding}, this is an overestimate to the true objective after performing the rounding procedure in Algorithm~\ref{alg:round}. In practice, we observed that this overestimate was negligibly different from the objective after rounding.

{\bf Sinkhorn implementation details.} We implemented the standard Sinkhorn algorithm, using different settings of $\eta^{-1}$. Sinkhorn iteration converges to an $\epsilon$-approximate transportation plan in theory when $\eta$ is very large, roughly $\log n / \epsilon$. However, in practice, it is observed that much smaller values of $\eta$ suffice for rapid convergence. We tracked the convergence of Sinkhorn iteration for $\eta = 70$ and $\eta = 5$, which we considered close to a theoretically guaranteed parameter and a much less conservative practical parameter, respectively. The optimized Sinkhorn algorithm converged at rates much faster than the predicted $\epsilon^{-2}$ rate on all experiments, outperforming all other methods, which we believe merits further investigation. Significantly larger values of $\eta$ led to numerical stability issues when computing $\exp(-\eta C)$.

{\bf APDAMD implementation details.} We implemented the APDAMD algorithm (Algorithm 4 in \cite{LinHJ19}), with the quadratic regularizer (i.e. $\frac{1}{2\gamma}\norm{\lambda}_2^2$). We observed that the amount of the quadratic regularizer added did not affect the practical convergence of the algorithm. A simple reason for this is because the algorithm builds in a more aggressive step-size strategy, because the pessimistic $\gamma = O(n)$ is often too conservative to be necessary in practice. The figure tracks APDAMD convergence with $\eta = 10^{-2}, \epsilon = 10^{-3}$.

{\bf Mirror prox.} For numerical stability considerations, we implemented our algorithm as an instance of mirror prox \cite{Nemirovski04}, another extragradient method which takes local iterations rather than accumulating a dual operator and taking steps with respect to some $\bar{z}$ (i.e. dual extrapolation). Although there is not a known proof of mirror prox convergence with an area-convex regularizer, we find this decision reasonable for several reasons. In general, variations of entropic mirror descent are well-known to be equivalent to their dual averaging versions; it is likely that a similar equivalence can be drawn between mirror prox and extragradient dual averaging, i.e. dual extrapolation. Furthermore, the standard proofs of dual extrapolation and mirror prox are quite similar; we believe it is likely that area-convexity results in convergence for mirror prox, although this merits further investigation.

{\bf Termination.} We terminated our alternating minimization procedure when the movement of iterations in $\ell_1$ was negligible. Typically, we observed that 3-5 alternating steps sufficed for convergence.

{\bf Step sizes.} We varied two parameters in our experiments: the step size $\frac{1}{\kappa}$ used in our extragradient algorithm, and the amount of entropy used in our regularizer (in the paper, we used 10 times entropy compared to the quadratic component $x^\top A^\top (y^2)$). One reason this may be reasonable in practice is similar to the observed behavior of the Sinkhorn iteration tuning the $\eta^{-1}$ parameter, and APDAMD performing a more-aggressive line search for the observed amount of regularizer necessary. To this end, we plotted the performance of three settings of our algorithm. 
\begin{itemize}
	\item In the ``unoptimized constants'', we set the constants to roughly those with theoretical guarantees, i.e. 10 times entropy and step size 1.
	\item In the ``reasonably optimized constants'', we set the amount of entropy to be 4, and the step size to be $\dmax / 3$, to offset the $\dmax$ multiple of the regularizer used in our iterations. For smaller values of $\epsilon$, these settings compared favorably with APDAMD.
	\item In the ``optimized constants'', we set the amount of entropy at 3, and the step size at $\dmax$. This setting outperformed APDAMD and was more competitive with Sinkhorn iteration.
\end{itemize}

{\bf Discussion.} We believe multiple interesting avenues of exploration arise from our experiments.
\begin{itemize}
	\item Sinkhorn with aggressively chosen $\eta$ outperformed all other methods we benchmarked against, and converged at rates faster than suggested by its known analyses. It may prove fruitful to study if further assumptions about practical instances explain this discrepancy.
	\item Directly accelerated methods such as APDAMD also exhibit $\epsilon^{-1}$ convergence rates, at the cost of a worse dependence on dimension. However, this worst-case dependence can be mitigated if the instance is favorable in practice, i.e. by choosing $\gamma \approx O(1)$. This was observed to be the case in our experiments for the MNIST dataset. It is interesting to see if a similar adaptive tuning applies to our method with provable guarantees.
	\item Our method did not exhibit instability when changing the amount of entropy in the regularizer, but it did exhibit vastly-improved convergence. It is possible that the amount of regularizer needed is not quite so large, perhaps through a more careful analysis.
	\item We did not benchmark against the greedy Sinkhorn method of \cite{AltschulerWR17}, or consider numerical speedups such as those in \cite{AltschulerBRW18}. It remains open to explore if these practical speedups are applicable to first-order methods such as ours as well.
\end{itemize}

	\end{appendix}
	
\end{document}